\documentclass{article}

\usepackage[colorlinks,linkcolor=magenta,citecolor=blue, pagebackref=true]{hyperref}
\usepackage[margin=1.2in]{geometry}
\usepackage[numbers]{natbib}
\usepackage{booktabs}
\usepackage{amssymb}
\usepackage{amsmath}
\usepackage{mathabx}

\usepackage[usenames,dvipsnames]{xcolor}

\DeclareFontFamily{U}{mathx}{\hyphenchar\font45}
\DeclareFontShape{U}{mathx}{m}{n}{<-> mathx10}{}
\DeclareSymbolFont{mathx}{U}{mathx}{m}{n}

\usepackage[utf8]{inputenc}
\usepackage{amsthm}
\usepackage{amsfonts}

\usepackage{dsfont}
\usepackage{graphicx}
\usepackage[makeroom]{cancel}
\usepackage[shortlabels]{enumitem}
\usepackage{thmtools}
\usepackage{algorithm}
\usepackage{algorithmic}
\usepackage[capitalise,noabbrev]{cleveref}
\usepackage{xcolor}

\declaretheorem[name=Theorem]{theorem}
\declaretheorem[name=Proposition]{proposition}
\declaretheorem[name=Lemma]{lemma}
\declaretheorem[name=Corollary]{corollary}

\theoremstyle{plain}
\theoremstyle{remark}

\newtheorem{remark}{Remark}

\newcommand{\PP}{\mathbb P}
\newcommand{\RR}{\mathbb R}
\newcommand{\EE}{\mathbb E}
\newcommand{\NN}{\mathbb N}
\newcommand{\ZZ}{\mathbb Z}

\newcommand{\Acal}{\mathcal A}

\newcommand{\Fcal}{\mathcal F}

\newcommand{\Hcal}{\mathcal H}

\newcommand{\Jcal}{\mathcal J}

\newcommand{\Ocal}{\mathcal O}
\newcommand{\Pcal}{\mathcal P}

\newcommand{\Xcal}{\mathcal X}

\newcommand{\1}{\mathds{1}}

\newcommand{\seq}[4]{(#1)_{{#2}={#3}}^{#4}}

\newcommand{\ci}{CI}

\newcommand{\cs}{CS}

\newcommand{\dd}{\ \mathrm{d}}

\DeclareMathOperator{\logit}{logit}
\newcommand{\indep}{\perp \!\!\! \perp}

\usepackage{setspace}

\newcommand{\history}{\mathcal{H}}
\newcommand{\algcomment}[1]{\textcolor{darkgray}{\textit{// #1}}}
\newcommand{\boundary}{\mathfrak B}
\newcommand{\DR}{\mathrm{DR}}
\newcommand{\IW}{\mathrm{IW}}
\newcommand{\DRL}{{(\mathrm{DR}\text{-}\ell)}}
\newcommand{\DRU}{{(\mathrm{DR}\text{-}u)}}
\newcommand{\IWL}{{(\mathrm{IW}\text{-}\ell)}}
\newcommand{\IWU}{{(\mathrm{IW}\text{-}u)}}
\newcommand{\LIL}{\mathrm{LIL}}
\newcommand{\EB}{{\mathrm{EB}}}

\newcommand{\wmax}{w_\mathrm{max}}
\newcommand{\infseqt}[1]{\seq{#1}{t}{1}{\infty}}
\newcommand{\PrPl}{{\mathrm{PrPl}}}
\newcommand{\onefone}{{_1F_1}}

\newcommand{\polval}{\nu}
\newcommand{\candpolval}{{\polval}'}
\newcommand{\polvalt}{\widetilde{ \polval}_t} 
\newcommand{\candpolvalt}{{\polvalt}'}
\newcommand{\poldiff}{\delta}
\newcommand{\poldifft}{\Delta_t}

\newcommand{\candpoldifft}{\Delta_t'}

\newcommand{\IA}{IA}

\DeclareMathOperator*{\essinf}{\mathrm{ess\ inf}}
\DeclareMathOperator*{\esssup}{\mathrm{ess\ sup}}

\usepackage{authblk}

\title{Anytime-valid off-policy inference for contextual bandits}
\author{
Ian Waudby-Smith$^1$, Lili Wu$^2$, Aaditya Ramdas$^1$, Nikos Karampatziakis$^2$, and Paul Mineiro$^2$\vspace{0.05in}\\
  $^1$Carnegie Mellon University\\
  $^2$Microsoft \vspace{0.05in}\\
  \texttt{ianws@cmu.edu}, \texttt{liliwu@microsoft.com}, \texttt{aramdas@cmu.edu},\\ \texttt{nikosk@microsoft.com}, \texttt{pmineiro@microsoft.com}
}
\date{}

\begin{document}
\maketitle
\setcounter{tocdepth}{2}
\makeatletter
\renewcommand\tableofcontents{%
  \@starttoc{toc}%
}

\makeatother

\begin{abstract}
Contextual bandit algorithms are ubiquitous tools for active sequential experimentation in healthcare and the tech industry. They involve online learning algorithms that adaptively learn policies over time to map observed contexts $X_t$ to actions $A_t$ in an attempt to maximize stochastic rewards $R_t$. This adaptivity raises interesting but hard statistical inference questions, especially counterfactual ones: for example, it is often of interest to estimate the properties of a hypothetical policy that is different from the logging policy that was used to collect the data --- a problem known as ``off-policy evaluation'' (OPE). 
Using modern martingale techniques, we present a comprehensive framework for OPE inference that relax unnecessary conditions made in some past works (such as performing inference at prespecified sample sizes, uniformly bounded importance weights, constant logging policies, constant policy values, among others), significantly improving on them both theoretically and empirically. Importantly, our methods can be employed while the original experiment is still running (that is, not necessarily post-hoc), when the logging policy may be itself changing (due to learning), and even if the context distributions are a highly dependent time-series (such as if they are drifting over time).
More concretely, we derive confidence sequences for various functionals of interest in OPE. These include doubly robust ones for time-varying off-policy mean reward values, but also confidence bands for the entire cumulative distribution function of the off-policy reward distribution. All of our methods (a) are valid at arbitrary stopping times (b) only make nonparametric assumptions, (c) do not require importance weights to be uniformly bounded and if they are, we do not need to know these bounds, and (d) adapt to the empirical variance of our estimators. In summary, our methods enable anytime-valid off-policy inference using adaptively collected contextual bandit data.
\end{abstract}

\tableofcontents

\section{Introduction}\label{section:introduction}

The so-called ``contextual bandit'' problem is an abstraction that can be used to describe several problem setups in statistics and machine learning \citep{langford2007epoch,li2010contextual}. For example, it generalizes the multi-armed bandit problem by allowing for ``contextual'' side information, and it can be used to describe many adaptive sequential experiments. The general contextual bandit problem can be described informally as follows: an agent (such as a medical scientist in a clinical trial) views contextual information $X_t \in \Xcal$ for subject $t$ (such as the clinical patient's demographics, medical history, etc.), takes an action $A_t \in \Acal$ (such as whether to administer a placebo, a low dose, or a high dose), and observes some reward $R_t$ (such as whether their adverse symptoms have subsided). This description is made formal in the protocol for the generation of contextual bandit data in \cref{algorithm:cb}. In the present paper, no restrictions are placed on the dimensionality or structure of the context and action spaces $\Xcal$ and $\Acal$ beyond them being measurable, but it is often helpful to think about $\Xcal$ as a $d$-dimensional Euclidean space, and $\Acal$ as $\{0, 1\}$ for binary treatments, or $\RR$ for different dosages, and so on. Indeed, while high-dimensional settings often pose certain challenges in contextual bandits (such as computational ones, or inflated variances), none of these issues will affect the \emph{validity} of our statistical inference methods. Throughout, we will require that the rewards are real-valued and bounded in $[0, 1]$ --- a common assumption in contextual bandits \citep{thomas2015high,karampatziakis2021off} --- except for \cref{section:cdf} where we relax the boundedness constraint.

There are two main objectives that one can study in the contextual bandit setup: (1) policy optimization, and (2) off-policy evaluation (OPE) \citep{li2010contextual,langford2007epoch,dudik2011doubly,dudik2014doubly}. Here, a ``policy'' $\pi(a \mid x)$ is simply a conditional distribution over actions, such as the probability that patient $t$ should receive various treatments given their context $X_t$. Policy \emph{optimization} is concerned with finding policies that achieve high cumulative rewards (typically measured through regret), while off-policy \emph{evaluation} is concerned with asking the counterfactual question: ``how would we have done if we used some policy $\pi$ instead of the policy that is currently collecting data?''. In this paper, we study the latter with a particular focus on statistical inference in adaptive, sequential environments under nonparametric assumptions.


\subsection{Off-policy inference, confidence intervals, and confidence sequences}
By far the most common parameter of interest in the OPE problem is the expected reward $\polval := \EE_{\pi}(R)$ that would result from taking an action from the policy $\pi$. This expectation $\polval$ is called the ``value'' of the policy $\pi$. While several estimators for $\polval$ have been derived and refined over the years, many practical problems call for more than just a point estimate: we may also wish to quantify the uncertainty surrounding our estimates via statistical inference tools such as confidence intervals (CI). However, a major drawback of \ci{}s is the fact that they are only valid at \emph{fixed and prespecified} sample sizes, while contextual bandit data are collected in a sequential and adaptive fashion over time.

We lay out the assumed protocol for the generation of contextual bandit data in~\cref{algorithm:cb}, and in particular, all of our results will assume access to the output of this algorithm, namely the (potentially infinite) sequence of tuples $(X_t, A_t, R_t)_{t=1}^T$ for $T \in \NN \cup \{ \infty\}$.
As is standard in OPE, we will always assume that the policy $\pi$ is (almost surely) absolutely continuous with respect to $h_t$ so that $\pi / h_t$ is almost surely finite (without which, estimation and inference are not possible in general). Indeed, this permits many bandit techniques and in principle allows for Thompson sampling since it always assigns positive probability to an action (note that it may not always easy to compute the probability of taking that action via Thompson sampling, but if those probabilities can be computed, they can be used directly within our framework). However, $\infseqt{h_t}$ cannot be the result of Upper Confidence Bound (UCB)-style algorithms since they take conditionally deterministic actions given the past, violating the absolute continuity of $\pi$ with respect to $h_t$.

In \cref{algorithm:cb}, the term ``exogenously time-varying'' simply means that the context and reward distributions at time $t$ can only depend on the past through $X_1^{t-1}  \equiv (X_1, \dots, X_{t-1})$, and not on the actions taken (or rewards received). Formally, we allow for any joint distribution over $\infseqt{X_t, A_t, R_t}$ as long as 
\begin{equation}\label{eq:conditional-independence}
    p_{R_t}(r \mid x, a, \history_{t-1}) = p_{R_t}(r \mid x, a, X_1^{t-1})~~~\text{and}~~~p_{X_t}(x \mid \history_{t-1}) = p_{X_t}(x \mid X_1^{t-1}),
\end{equation}
where $\history_t$ is all of the history $\sigma\left ((X_i, A_i, R_i)_{i=1}^t \right )$ up until time $t$. This conditional independence requirement \eqref{eq:conditional-independence} includes as a special case more classical setups where $X_t$ is independent of all else given $A_t$, such as those considered in \citet{bibaut2021post} or iid scenarios \citep{karampatziakis2021off}, but is strictly more general, since, for example, $\infseqt{X_t}$ can be a highly dependent time-series. However, we do not go as far as to consider the adversarial setting that is sometimes studied in the context of regret minimization. We impose this conditional independence requirement since otherwise, the interpretation of $\EE_{\pi}(R_t \mid \history_{t-1})$ changes depending on which sequence of actions were played by the logging policy. Making matters more concrete, the conditional off-policy value $\EE_\pi(R_t \mid \history_{t-1})$ at time $t$ is given by
\begin{align}
    \polval_t := \EE_{\pi} (R_t \mid \history_{t-1}) &\equiv \int_{\Xcal \times \Acal \times \RR} r\cdot p_{R_t}(r \mid a, x, \history_{t-1}) \pi(a \mid x) p_{X_t}(x \mid \history_{t-1}) \dd x \dd a \dd r \label{eq:before-pseudo-identification}\\
    &=\int_{\Xcal \times \Acal \times \RR} r\cdot p_{R_t}(r \mid a, x, X_1^{t-1}) \pi(a \mid x) p_{X_t}(x \mid X_1^{t-1}) \dd x \dd a \dd r,\label{eq:pseudo-identification}
\end{align}
and the equality \eqref{eq:pseudo-identification} follows from \eqref{eq:conditional-independence}. Notice that \eqref{eq:before-pseudo-identification} could in principle depend on the logging policies and actions played, but \eqref{eq:pseudo-identification} does not. Despite imposing the conditional independence assumption \eqref{eq:conditional-independence}, the integral \eqref{eq:before-pseudo-identification} is still a perfectly well-defined functional, and if \eqref{eq:conditional-independence} is not satisfied, then our \cs{}s will still cover a quantity in terms of this functional. However, its interpretation would no longer be counterfactual with respect to the entire sequence of actions (only conditional on the past).

While most prior work on OPE in contextual bandits is not written \emph{causally} in terms of potential outcomes (e.g.~\citep{thomas2015high,karampatziakis2021off,bibaut2021post,chandak2021universal,huang2021off}), it is nevertheless possible to write down a causal target $\nu^\star_t$ (i.e.~a functional of the potential outcome distribution) and show that it is equal to $\polval_t$ under certain causal identification assumptions. These assumptions resemble the familiar \emph{consistency}, \emph{exchangeability}, and \emph{positivity} conditions that are ubiquitous in the treatment effect estimation literature. Moreover, there is a close relationship between OPE and the estimation of so-called \emph{stochastic interventions} in causal inference; indeed, they can essentially be seen as equivalent but with slightly different emphases and setups. However, given that neither the potential outcomes view nor the stochastic intervention interpretation of OPE are typically emphasized in the contextual bandit literature (with the exception of \citet{zhan2021off}, who use potential outcomes throughout), we leave this discussion to~\cref{section:causal}. 

\begin{algorithm}[!htbp]
  \caption{Protocol for the generation of contextual bandit data}
  \label{algorithm:cb}
  \begin{algorithmic}
    \STATE{\algcomment{Here, $T \in \NN \cup \{\infty\}$.} }
    \FOR{$t = 1,2, \dots, T$}
    \STATE{\algcomment{The agent selects a policy $h_t$ based on the history $\history_{t-1} \equiv \sigma \left ((X_i, A_i, R_i)_{i=1}^{t-1} \right)$.}}
    \STATE{$h_t \in \history_{t-1}$.}\;
    \STATE{}\;
    \STATE{\algcomment{The environment draws a context from an (exogenously time-varying) distribution.}}\;
    \STATE{$X_t \sim p_{X_t}(\cdot)$}\;
    \STATE{}\;
    \STATE{\algcomment{The agent plays a random action drawn from the selected policy.}}\;
    \STATE{$A_t \sim h_t(\cdot \mid X_t)$.}\;
    \STATE{}\;
    \STATE{\algcomment{The environment draws a reward from an (exogenously time-varying) distribution based on the action and context.}}\;
    \STATE{$R_t \sim p_{R_t}(\cdot \mid A_t, X_t)$.}\;
    \ENDFOR{}

    \STATE{} \;
    \STATE{\algcomment{Return a (potentially infinite) sequence of contextual bandit data.}}
    \RETURN{$(X_t, A_t, R_t)_{t=1}^T$}
  \end{algorithmic}
\end{algorithm}


To illustrate the shortcomings of \ci{}s for OPE, suppose we run a contextual bandit algorithm and want to see whether $\pi$ is better than the current state-of-the-art policy $h$ --- e.g.~whether $\EE_{\pi}(R) > \EE_{h}(R)$. (Here, we are implicitly assuming that $\EE_{\pi'}(R) = \EE_{\pi'}(R_t \mid \history_{t-1})$ for any policy $\pi'$ for the sake of illustration.) Suppose we compute a \ci{} for the value of $\pi$ based on $n$ samples (for some prespecified $n$), and while $\pi$ seems promising, the \ci{} turns out to be inconclusive (the \ci{} for $\EE_{\pi}(R)$ includes $\EE_{h}(R)$ if the latter is known, or the two \ci{}s overlap if the latter is unknown). It is tempting to  collect more data, for a total of $n'$ points, to see if the result is now conclusive; however the resulting sample size $n'$ is now a \emph{data-dependent} quantity, rendering the \ci{} invalid. (This could happen more than once.)

Fortunately, there exist statistical tools that permit adaptive stopping in these types of sequential data collection problems: \emph{confidence sequences} (CSs~\cite{darling1967confidence,lai1976confidence}). A \cs{} is a sequence of confidence intervals, valid at all sample sizes uniformly (and hence at arbitrary stopping times). Importantly for the aforementioned OPE setup, \cs{}s allow practitioners to collect additional data and continuously monitor it, so that the resulting \ci{} is indeed valid at the data-dependent stopped sample size $n'$. More formally, we say that a sequence of intervals $[L_t, U_t]_{t=1}^\infty$ is a \cs{} for the parameter $\theta \in \RR$ if
\begin{equation}\label{eq:cs}
  \PP \left ( \forall t \in \NN,\ \theta \in [L_t, U_t] \right ) \geq 1-\alpha,
  ~~\text{or equivalently,}~~~ \PP \left ( \exists t \in \NN : \theta \notin [L_t, U_t] \right ) \leq \alpha.
\end{equation}
Contrast~\eqref{eq:cs} above with the definition of a \ci{} which states that $\forall n \in \NN,\ \PP(\theta \in [L_n, U_n]) \geq 1-\alpha$, so that the ``$\forall n$'' is outside the probability $\PP(\cdot)$ rather than inside. A powerful consequence of~\eqref{eq:cs} is that $[L_\tau, U_\tau]$ is a valid \ci{} for \emph{any stopping time} $\tau$. In fact, $[L_\tau, U_\tau]$ being a valid \ci{} is not just an implication of \eqref{eq:cs} but the two statements are actually equivalent; see \citet[Lemma 3]{howard2018uniform}. 

The consequence for the OPE practitioner is that they can continuously update and monitor a \cs{} for the value of $\pi$ \emph{while} the contextual bandit algorithm is running, and deploy $\pi$ as soon as they are confident that it is better than the current state-of-the-art $h$. \citet{karampatziakis2021off} refer to this adaptive policy switching as ``gated deployment'', and we will return to this motivating example through the paper. Let us now lay out five desiderata that we want all of our \cs{}s to satisfy.

\subsection{Desiderata for anytime-valid off-policy inference}\label{section:desiderata}
Throughout this paper, we will derive methods for off-policy evaluation and inference in a variety of settings --- including fixed policies (\cref{section:warmup}), time-varying policies (\cref{section:time-varying}), and for entire cumulative distribution functions (\cref{section:cdf}). However, what all of these approaches will have in common is that they will satisfy five desirable properties which we enumerate here.

\begin{enumerate}
\item \textbf{Nonasymptotic:} Our confidence sets will satisfy \emph{exact} coverage guarantees for \emph{any} sample size, unlike procedures based on the central limit theorem which only satisfy \emph{approximate} guarantees for large samples.\footnote{Note that nonasymptotic procedures may be more conservative than asymptotic ones as they satisfy more rigorous coverage (similarly, type-I error) guarantees. Whether one should sacrifice stronger guarantees for tightness comes down to philosophical preference. For the purposes of this paper, we focus solely on nonasymptotics.}
\item \textbf{Nonparametric:} We will not make any parametric assumptions on the distribution of the contexts, policies, or rewards.
\item \textbf{Time-uniform / anytime-valid:} Our confidence sets will be \emph{uniformly} valid for all sample sizes, and permit off-policy inference at arbitrary data-dependent stopping times.
\item \textbf{Adaptive data collection (via online learning):} All of our off-policy inference methods will allow for the sequence of logging policies $(h_t)_{t=1}^\infty$ to be predictable (i.e.~$h_t$ can depend on $\Hcal_{t-1}$). In particular $\infseqt {h_t}$ can be the result of an online learning algorithm.
\item \textbf{Unknown and unbounded $\wmax$:} In all of our algorithms, the maximal importance weight
  \[\wmax := \esssup_{t\in \NN, a \in \Acal, x \in \Xcal} \frac{\pi(a \mid x)}{h_t(a \mid x)} \]
  can be unknown, and need not be uniformly bounded (i.e.~$\wmax$ can be infinite). Note that we do require that importance weights $\pi(a\mid x) / h_t(a \mid x)$ themselves are finite for each $(t, a, x)$, but their essential \emph{supremum} need not be. Perhaps surprisingly, even if $\wmax$ is infinite, it is still possible for our \cs{}s to shrink to zero-width since they depend only on \emph{empirical variances}. As an illustrative example, see \cref{proposition:lil-eb} for a closed-form \cs{} whose width can shrink at a rate of $\sqrt{\log\log t / t}$ as long as the importance-weighted rewards are well-behaved (e.g.~in the iid setting, if they have finite second moments).
\end{enumerate}
In addition to the above, we will design procedures so that they have strong empirical performance and are straightforward to compute. While some of these desiderata are quite intuitive and common in statistical inference (such as nonasymptotic, nonparametric, and time-uniform validity, so as to avoid relying on large sample theory, unrealistic parametric assumptions, or prespecified sample sizes), desiderata 4 and 5 are more specific to OPE and have not been satisfied in several prior works as we outline in Sections~\ref{section:warmup},~\ref{section:time-varying}, and~\ref{section:cdf}. Given their central importance to our work, let us elaborate on them here.

\paragraph{Why allow for logging policies to be predictable?} For the purpose of policy optimization, contextual bandit algorithms are tasked with balancing exploration and exploitation: simultaneously figuring out which policies will yield high rewards (at the expense of trying out suboptimal policies) and playing actions from the policies that have proven effective so far. On the other hand, in adaptive sequential trials, an experimenter might aim to balance context distributions between treatment arms (such as via Efron's biased coin design \citep{efron1971forcing}) or to adaptively find the treatment policy that yields the most efficient treatment effect estimators~\citep{kato2020adaptive}. In both cases, the logging policies $\infseqt{h_t}$ are not only changing over time, but adaptively so based on previous observations. We strive to design procedures that permit inference in precisely these types of adaptive data collection regimes, despite most prior works on off-policy inference for contextual bandits having assumed that there is a fixed, prespecified logging policy \citep{thomas2015high,karampatziakis2021off,chandak2021universal,huang2021off}. Of course, if a \cs{} or \ci{} is valid under adaptive data collection, they are also valid when fixed logging policies are used instead.

\paragraph{Why not rely on knowledge of $\wmax$?} Related to the previous discussion, it may not be known \emph{a priori} how the range of the predictable logging policies will evolve over time. Moreover, one can imagine a situation where $\sup_{a, x} \pi(a \mid x ) / h_t(a \mid x) \to \infty$, even if every individual importance weight $\pi / h_t$ is finite. In such cases, having \cs{}s be agnostic to the value of $\wmax$ is essential. However, even if $\wmax$ \emph{is} known, it may be preferable to design \cs{}s that do not depend on this worst-case value. Suppose for the sake of illustration that a logging policy $h$ assigns a novel treatment (denoted by $a = 1$) with probability 1/5 and a placebo (denoted by $a = 0$) with probability 4/5 for most subjects, except for a small but high-risk subpopulation, who receive the novel treatment with probability 1/1000. To estimate the expected reward of the novel treatment, note that the importance weight for subject $t$ will take on the value $w_t := 1 / h(A_t \mid X_t)$ for treatment $A_t \in \{0, 1\}$ and context $X_t \in \Xcal$. Despite the fact that most of the importance weights are only 5, and hence most importance-weighted pseudo-outcomes $w_t R_t$ will take values in $[0, 5]$, the worst-case $\wmax$ is much larger at 1000. Consequently, we should expect a \cs{} that scales with $\wmax = 1000$ to be much wider than one that only scales with a quantity like an empirical variance. For these reasons we prefer procedures that depend on an empirical variance term (defined later) rather than the worst-case importance weight $\wmax$.

\subsection{Outline and contributions}
Our fundamental contribution is in the derivation of \cs{}s for various off-policy parameters, including fixed policy values, time-varying policy values, and quantiles of the off-policy reward distribution. We begin in~\cref{section:warmup} with the most common formulation of the OPE problem: estimating the value $\polval$ of a target policy $\pi$. \cref{theorem:dr-fixed-policy-value} presents time-uniform \cs{}s for $\polval$, a result that generalizes and improves upon the current state-of-the-art \cs{}s for $\polval$  by \citet{karampatziakis2021off}. In~\cref{section:time-varying}, we consider the more challenging problem of estimating a time-varying average policy value $\polval_t$, where the distribution of off-policy rewards can change over time in an arbitrary and unknown fashion. In \cref{section:cdf}, we derive \cs{}s for quantiles of the off-policy reward distribution, and in particular, \cref{theorem:cdf} presents a confidence band for the entire cumulative distribution function (CDF) that is both uniformly valid in time \emph{and} in the quantiles. For the results of Sections~\ref{section:time-varying} and~\ref{section:cdf}, no other solutions to this problem exist in the literature, to the best of our knowledge. Finally, in \cref{section:summary}, we summarize our results and describe some natural extensions and implications of them, namely false discovery rate control under arbitrary dependence when evaluating several policies, and differentially private off-policy inference.

\subsection{Related work}

Throughout the paper, we will draw detailed comparisons to work that is most closely related to ours, i.e.~papers that are broadly concerned with estimating policy values and/or the off-policy CDF from contextual bandit data in a model-free setting --- here we are using the term ``model-free'' to mean that no restrictions are placed on the functional form between the rewards $R$ and the actions $A$ nor on the covariates $X$. Specifically, \cref{table:time-varying} and the preceding text provides a (selective) property-by-property comparison to the directly related works of \citet{karampatziakis2021off,bibaut2021post,zhan2021off}, and \citet{howard2018uniform}, and \cref{table:cdf} provides a similar comparison to the works of \citet{howard2022sequential}, \citet{chandak2021universal}, and \citet{huang2021off}. However, there are several other works that focus on \emph{fixed-$n$} (i.e.~not time-uniform) and \emph{asymptotic} statistical inference from adaptively collected data (e.g.~in the form of multi-armed bandits, contextual bandits, or more general reinforcement learning). For example, \citet{ramprasad2022online} develop a bootstrap procedure for estimating policy values under Markov noise with temporal difference learning algorithms, \citet{dimakopoulou2021online} perform adaptive inference in the multi-armed bandit setting, \citet{hadad2021confidence} provide asymptotic confidence intervals for treatment effects in adaptive experiments, and \citet{zhang2021statistical} provide distribution-uniform asymptotic procedures for M-estimation from contextual bandit data. Other works that consider more model-based approaches include  \citet{zhang2020inference}, \citet{khamaru2021near}, \citet{shen2021doubly}, and \citet{chen2021statistical}.

\subsection{Notation: supermartingales, filtrations, and stopping times}

Since all of our results will rely on the analysis of nonnegative (super)martingales, predictable processes, stopping times, and so on, it is worth defining some of these terms before proceeding. Consider a universe of distributions $\Pi$ on a filtered probability space $(\Omega,\Fcal)$. A single draw from any distribution $P \in \Pi$ results in a sequence $Z_1,Z_2,\dots$ of potentially dependent observations. (In the context of this paper, $Z_t$ may represent $(X_t,A_t,R_t)$, for example, and the distribution $P$ may be induced by the policy, and not specified in advance.) If $Z_1, Z_2, \dots$ are independent and identically distributed (iid), we will explicitly say so, but in general we eschew iid assumptions in this paper.

As is common in the statistics literature, we will use upper-case letters like $Z$ to refer to random variables and lower-case letters $z$ to refer to non-stochastic values in the same space that $Z$ takes values.
Let $Z_1^t$ denote the tuple $(Z_1,\dots,Z_t)$ and let $\history \equiv \infseqt{\history_t}$ by default represent the data (or ``canonical'') filtration, meaning that $\Hcal_t=\sigma(Z_1^t)$.

A sequence of random variables $Y \equiv \infseqt{Y_t}$ is called a \textit{process} if it is adapted to $\Hcal$, that is if $Y_t$ is measurable with respect to $\Hcal_t$ for every $t$. A process $Y$ is \emph{predictable} if $Y_t$ is measurable with respect to $\history_{t-1}$ --- informally ``$Y_t$ only depends on the past''. A process $M$ is a \emph{martingale} for $P$ with respect to $\history$ if \begin{equation}\label{eq:martingale}
    \mathbb{E}_P[M_t \mid \Hcal_{t-1}] = M_{t-1} 
\end{equation}
for all $t\geq 1$. $M$ is a \emph{supermartingale} for $P$ if it satisfies \eqref{eq:martingale} with ``$=$'' relaxed to ``$\leq$''.  A (super)martingale is called a \emph{test (super)martingale} if it is nonnegative and $M_0=1$. A process $M$ is called a test (super)martingale for $\Pcal \subset \Pi$ if it is a test (super)martingale for every $P \in \Pcal$. 

Throughout, if an expectation $\EE$ operator is used without a subscript $P$, or if a \textbf{boldface} $\PP$ is used to denote a probability, these are always referring to the  distribution of $(X_t,A_t,R_t)_{t \geq 1}$ induced by \cref{algorithm:cb} and the logging policies $\infseqt{h_t}$.

An $\Hcal$-stopping time $\tau$ is a $\NN \cup \{\infty\}$-valued random variable such that $\{\tau \leq t\} \in \Hcal_t$ for each $t \geq 0$. Informally and in the context of this paper, a stopping time can be thought of as a sample size that was chosen based on all of the information $\history_t$ up until time $t$.

\section{Warmup: Off-policy inference for constant policy values}\label{section:warmup}

This section deals with the case where $\nu_t$ from~\eqref{eq:before-pseudo-identification} does not depend on $t$, meaning that it is constant as a function of time. We handle the time-varying case in the next section.

We begin by extending a result of \citet[Section 5.2]{karampatziakis2021off} which applied in the iid setting, meaning that the logging policy $h$ is fixed and the contexts and rewards are assumed to be iid. Their paper derives several \cs{}s for the value $\polval := \EE_{A \sim \pi}(R)$ of the policy $\pi$ for $[0, 1]$-bounded
rewards $R$, but some of their \cs{}s require knowledge of $\wmax$, which we would like to avoid as per our desiderata in \cref{section:desiderata}. However, their so-called ``scalar betting'' approach in~\citep[Section 5.2]{karampatziakis2021off} makes use of importance-weighted random variables and does not depend on knowing $\wmax$. To elaborate, let $w_t$ be the importance weight for the target policy $\pi$ versus the logging policy $h$ given by
  \begin{equation}
    w_t := \frac{\pi(A_t \mid X_t)}{h(A_t \mid X_t)},
  \end{equation}
  and let $\phi_t^\IWL := w_t R_t$ and $\phi_t^\IWU := w_t(1-R_t)$ be importance-weighted rewards that will be used to construct lower and upper bounds respectively. Note that $\phi_t^\IWL$ is ubiquitous in the bandit and causal inference literatures, and the authors were not concerned with deriving \emph{new estimators}, but rather \emph{new confidence sequences} using existing estimators. While $w_t \equiv w_t(X_t, A_t)$ does depend on both $A_t$ and $X_t$, we leave the dependence on them implicit going forward to reduce notational clutter.

\begin{proposition}[Scalar betting off-policy \cs{} \cite{karampatziakis2021off}]\label{proposition:kmr}
  Suppose $(X_t, A_t, R_t)_{t=1}^\infty$ are iid with $[0, 1]$-valued rewards $\infseqt{R_t}$, and the logging policy $h$ is fixed. For each $\candpolval \in [0, 1]$, let   $(\lambda_t^L(\candpolval))_{t=1}^\infty$ be any $[0, 1/\candpolval)$-valued predictable sequence.
  Then,
  \begin{equation}\label{eq:kmr-cs}
    L_t^\IW := \inf\left \{ \candpolval \in [0, 1] : \prod_{i=1}^t \left ( 1 + \lambda_i^L(\candpolval) \cdot (\phi_t^\IWL - \candpolval) \right ) < \frac{1}{\alpha} \right \}
  \end{equation}
  forms a lower $(1-\alpha)$-\cs{} for $\polval$, meaning $\PP(\forall t \in \NN,\ \polval \geq L_t^\IW) \geq 1-\alpha$. Similarly, for any $[0, 1/(1-\candpolval))$-valued predictable sequence $\infseqt{\lambda_t^U(\candpolval)}$.
  \begin{equation}
    U_t^\IW := 1 - \inf \left \{ 1-\candpolval \in [0, 1] : \prod_{i=1}^t \left [ 1 + \lambda_i^U(\candpolval) \cdot (\phi_t^\IWU - (1-\candpolval)) \right ] < \frac{1}{\alpha} \right \}
  \end{equation}
  forms an upper $(1-\alpha)$-\cs{} for $\polval$, meaning $\PP \left ( \forall t \in \NN,\ \polval \leq U_t^\IW \right ) \geq 1-\alpha$. A two-sided \cs{} can be formed using $[L_t^\IW, U_t^\IW]_{t=1}^\infty$ combined with a union bound.
\end{proposition}
The above \cs{} $[L_t^\IW, U_t^\IW]_{t=1}^\infty$ due to \citet{karampatziakis2021off} has a number of desirable properties. Namely, it satisfies the first four of five desiderata in \cref{section:desiderata}, meaning it is a nonasymptotic, nonparametric, time-uniform confidence sequence that does not require knowledge of $\wmax$. Note that while infima appear in the definitions of $L_t^\IW$ and $U_t^\IW$ they are straightforward to compute (e.g.~via line or grid search) when the product is quasiconvex in $\candpolval \in [0, 1]$ which is often the case as we discuss in \cref{section:tuning-truncating-mirroring}.
The idea behind \cref{proposition:kmr} is to show that the product inside the above infima are nonnegative martingales when $\candpolval = \polval$ and then apply Ville's inequality to it \citep{ville1939etude}. Our main results in the coming sections use similar techniques albeit with very different (super)martingales tailored to different problem settings.

We also wish to highlight that indeed, $[L_t^\IW, U_t^\IW]_{t=1}^\infty$ forms a valid $(1-\alpha)$-\cs{} \emph{regardless} of how the sequences $(\lambda_t^L(\candpolval))_t$ and $(\lambda_t^U(\candpolval))_t$ are chosen. Such phenomena are common in martingale-based statistical procedures such as in \citet{waudby2020estimating} (see also the review paper of \citet{ramdas2022game}) and will be seen in several of the results to follow. We will discuss some guiding principles for how to choose these sequences in \cref{remark:choosing-lambda}.

\subsection{Tighter confidence sequences via doubly robust pseudo-outcomes}\label{section:dr}

Here, we generalize and improve upon~\cref{proposition:kmr} in three ways. First, we show that the logging policy $h$ can be replaced by a \emph{sequence} of predictable logging policies $(h_t)_{t=1}^\infty$ so that $h_t$ can be built from the entire history $\Hcal_{t-1}$ up until time $t-1$ (and in particular, $\infseqt{h_t}$ can be the result of an online learning algorithm), so that the importance weight $w_t$ at time $t$ is given by
\begin{equation}\label{eq:predictable-importance-weights}
  w_t := \frac{\pi(A_t \mid X_t)}{h_t(A_t \mid X_t)}.
\end{equation}

Second, we show how the importance-weighted pseudo-outcomes $(\phi_t)_{t=1}^\infty$ can be made doubly robust in the sense of \citet{dudik2011doubly,dudik2014doubly}. Indeed, define the lower and upper doubly robust pseudo-outcomes $(\phi_t^\DRL)_{t=1}^\infty$ and $(\phi_t^\DRU)_{t=1}^\infty$ given by
\begin{align}
  \phi_t^\DRL &:= w_t \cdot \left (R_t - \left [ \widehat r_t(X_t; A_t) \land \frac{k_t}{w_t} \right ] \right ) + \EE_{a \sim \pi(\cdot \mid X_t)} \left ( \widehat r_t(X_t; a) \land \frac{k_t}{w_t} \right ) , \label{eq:dr-outcomes-lower}\\
  \phi_t^\DRU &:= w_t \cdot \left (1-R_t - \left [ (1-\widehat r_t(X_t; A_t)) \land \frac{k_t}{w_t} \right ] \right ) + \EE_{a \sim \pi(\cdot \mid X_t)} \left ( [1-\widehat r_t(X_t; a) ]\land \frac{k_t}{w_t} \right ) , \label{eq:dr-outcomes-upper}
\end{align}
where $\widehat r_t(X_t ; A_t)$ is any $[0, 1]$-valued predictor of $R_t$ built from $\history_{t-1}$ and $k_t$ is a $\RR_{\geq 0}\cup \{\infty\}$-valued tuning parameter built from $\history_{t-1}$ that determines how ``doubly robust'' $\phi_t^\DR$ should be. Note that $\phi_t^\DRL$ and $\phi_t^\DRU$ are both at least $-k_t$ by construction, and have conditional means of $\nu$ and $1-\nu$ respectively. Note that the phrase ``doubly robust'' is sometimes used to refer to \textit{properties} of estimators (e.g.~that their bias is second order and depends only on products of nuisance errors in observational studies where importance weights are unknown \citep{kennedy2022semiparametric}) and sometimes to refer to \emph{types of} estimators that enjoy variance-reduction without compromising validity in experiments where importance weights are known. We are using this phrase in the second sense following the conventions of \citet{dudik2011doubly,dudik2014doubly}. 

Similar to the discussion surrounding \cref{proposition:kmr}, the doubly robust pseudo-outcomes in \eqref{eq:dr-outcomes-lower} are ubiquitous in the causal inference and bandit literatures \citep{robins1994estimation,dudik2011doubly,dudik2014doubly,uehara2022review} with the minor tweak that we are truncating the reward predictor. Note that we are \emph{not} doing this for the purposes of deriving better estimators --- instead, we are doing so for the purposes of sharp concentration of measure in the pursuit of tighter CSs.

Setting $k_1 = k_2 = \cdots = 0$ recovers the IW outcomes exactly, while setting $k_1 = k_2 = \cdots = \infty$ recovers the classic doubly robust outcomes \citep{dudik2011doubly,dudik2014doubly} (this could also be achieved by setting $k_1 = k_2= \cdots = \wmax$, provided $\wmax$ is finite and known). We discuss the need to truncate $\widehat r_t$ in \cref{remark:DR-why-truncate}, but the motivation to include a reward predictor at all is to reduce the variance of $\phi_t^\DRL$ and $\phi_t^\DRU$ if $R_t$ can be well-predicted by $\widehat r_t$, a well-known phenomenon in doubly robust estimation \citep{robins1994estimation,van2011targeted,chernozhukov2017double}. Indeed, we find that the resulting \cs{}s are able to adapt to this reduced variance accordingly for large $t$ (see \cref{fig:iw-vs-dr} for an illustration).
\begin{figure}[!htbp]
  \centering
  Confidence sequences and their widths when the policy value is $\nu = 0.6$\\
  \includegraphics[width=\textwidth]{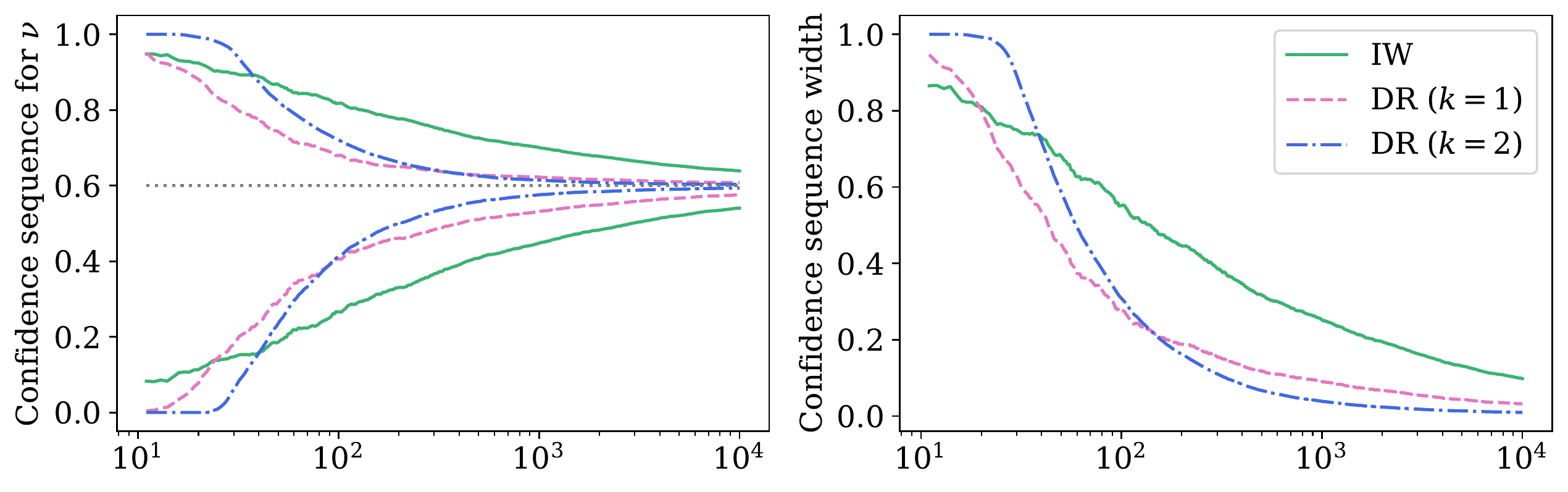}
  Confidence sequences and their widths when the policy value is $\nu = 0.1$\\
  \includegraphics[width=\textwidth]{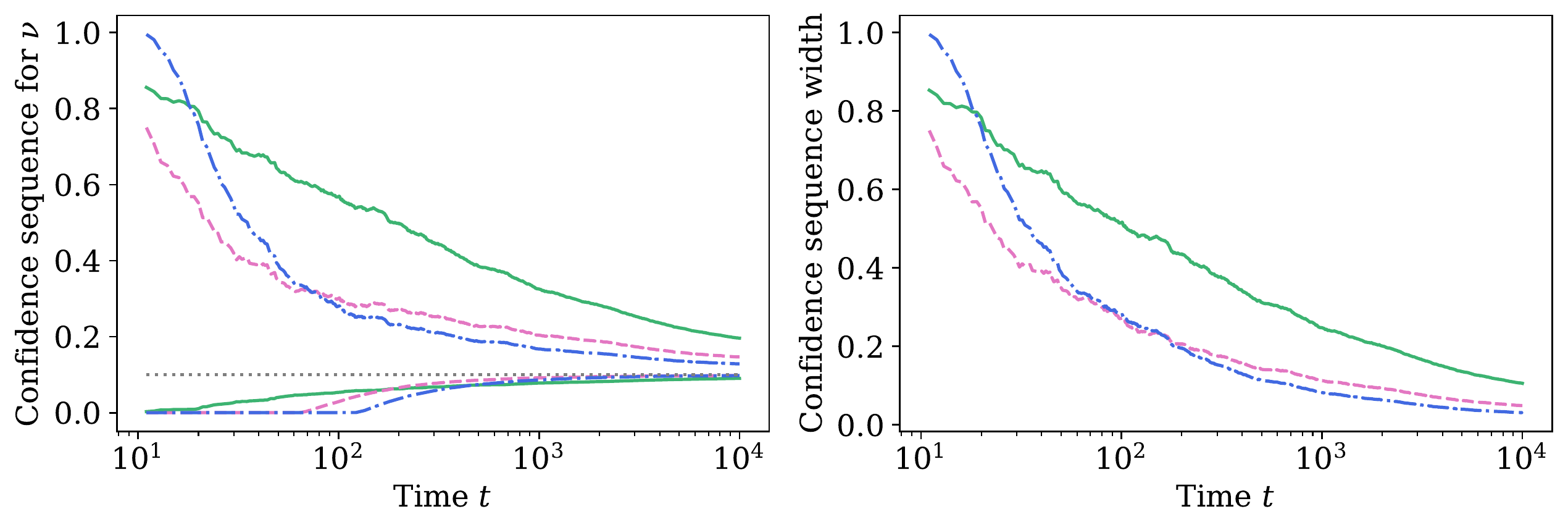}
  \caption{Three confidence sequences for a policies with values $\polval = 0.6$ and $\polval = 0.1$. The first CS is built from importance-weighted pseudo-outcomes (``IW''), and the other two are built from doubly robust pseudo-outcomes (``DR'') with $k$ taking values 1 and 2, respectively. In these examples, the reward $R_t$ can be predicted easily, a property that only the doubly robust \cs{}s can exploit. Notice that a larger value of $k$ allows the doubly robust \cs{} to become narrower for large $t$, but it pays for this adaptivity with wider bounds at small $t$. Nevertheless, all three \cs{}s are time-uniform, and nonasymptotically valid in both simulation scenarios.}
  \label{fig:iw-vs-dr}
\end{figure}

Third and finally, we relax the iid assumption, and only require that $\EE_{\pi}(R_t \mid \Hcal_{t-1}) = \polval \equiv \EE_\pi (R_t) $ and $R_t \in [0, 1]$ almost surely. This relaxation of assumptions can be obtained for free, without any change to the resulting \cs{}s whatsoever, and with only a slight modification to the proof. We summarize our extensions in the following theorem.

\begin{theorem}[Doubly robust betting off-policy \cs{}]\label{theorem:dr-fixed-policy-value}
 Suppose $(X_t, A_t, R_t)_{t=1}^\infty$ is an infinite sequence of contextual bandit data generated by the predictable policies $(h_t)_{t=1}^\infty$ whose $[0, 1]$-valued reward $R_t$ at time $t$ has conditional mean $\EE_\pi(R_t \mid \history_{t-1}) = \polval \equiv \EE_\pi (R_t)$. 
 For any predictable sequence $(\lambda_t^L(\candpolval))_{t=1}^\infty$ such that $\lambda_t^L(\candpolval) \in [0, (\candpolval + k_t)^{-1})$, we have
  \begin{equation}\label{eq:dr-lower}
    L_t^\DR := \inf \left \{ \candpolval \in [0, 1] : \prod_{i=1}^t \left [ 1 + \lambda_i^L(\candpolval) \cdot (\phi_i^\DRL - \candpolval) \right ] < \frac{1}{\alpha} \right \}
  \end{equation}
  forms a lower $(1-\alpha)$-\cs{} for $\polval$. Similarly, if $\lambda_t^U(\candpolval) \in \left [0, (1-\candpolval + k_t)^{-1}\right )$ is predictable, then
  \begin{equation}\label{eq:dr-upper}
    U_t^\DR := 1-\inf \left \{ 1-\candpolval \in [0, 1] : \prod_{i=1}^t \left [ 1 + \lambda_i^U(\candpolval) \cdot (\phi_i^\DRU - (1-\candpolval)) \right ] < \frac{1}{\alpha} \right \}
  \end{equation}
  forms an upper $(1-\alpha)$-\cs{} for $\polval$. 
\end{theorem}
The proof of \cref{theorem:dr-fixed-policy-value} can be found in \cref{proof:dr-fixed-policy-value} and relies on applying Ville's inequality~\citep{ville1939etude} to the products in \eqref{eq:dr-lower} and \eqref{eq:dr-upper}. Note that the dimensionality of $\Xcal$ does \emph{not} in any way affect the validity of \cref{theorem:dr-fixed-policy-value}. Moreover, $\phi_t^\DRL$ and $\phi_t^\DRU$ are always unbiased and yield valid \cs{}s regardless of how $\widehat r_t$ is chosen. The reason to introduce these doubly robust pseudo-outcomes is to obtain lower-variance CSs (as illustrated in \cref{fig:iw-vs-dr}) since doubly robust estimators can be semiparametric efficient thereby attaining the optimal asymptotic mean squared error in a local minimax sense. These details are outside the scope of the present paper, but we direct the interested reader to \citet{kennedy2022semiparametric} and \citet{uehara2022review} for modern reviews discussing this subject.

Notice that \cref{theorem:dr-fixed-policy-value} is a generalization of \cref{proposition:kmr}. Indeed, if the logging policies do not change (i.e.~$h_1=h_2=\cdots=h$), and if the observations $\infseqt{X_t, A_t, R_t}$ are iid, and if $k_t = 0$ for each $t$, then \cref{theorem:dr-fixed-policy-value} recovers \cref{proposition:kmr} exactly. For this reason, we do not elaborate on empirical comparisons between~\cref{proposition:kmr} and \cref{theorem:dr-fixed-policy-value} --- any \cs{} that can be derived using the former is a special case of the latter. Moreover, in the \emph{on-policy} setting with all importance weights set to 1 and without a reward predictor, \cref{theorem:dr-fixed-policy-value} recovers the betting-style \cs{}s of \citet[Theorem 3]{waudby2020estimating}. As alluded to in the discussion following \cref{proposition:kmr}, the infima above are straightforward to compute for many choices of predictable sequences $\infseqt{\lambda_t^L(\candpolval)}$ and $\infseqt{\lambda_t^U(\candpolval)}$ including all of those discussed in the following section.

\subsection{Tuning, truncating, and mirroring}\label{section:tuning-truncating-mirroring}

We make three remarks below, that are important on both theoretical and practical fronts.

\begin{remark}[Tuning $\infseqt{\lambda_t^L}$ and $\infseqt{k_t}$]\label{remark:choosing-lambda}
  As stated, \cref{theorem:dr-fixed-policy-value} yields a valid lower \cs{} for $\polval$ using \emph{any} predictable sequence of $[0, (\candpolval + k_t)^{-1})$-valued tuning parameters $(\lambda_t^L(\candpolval))$ --- referred to as ``bets'' by \citet{karampatziakis2021off} and \citet{waudby2020estimating}, but how should these bets be chosen? \citet[Appendix B]{waudby2020estimating} discuss several possible options, but in practice none of them uniformly dominate the others. (This should not be surprising, since there is a certain formal sense in which different nontrivial nonnegative martingales cannot uniformly dominate each other for a given composite sequential testing problem; see~\citet{ramdas2020admissible} for a precise statement.) For a simple-to-implement and empirically compelling option, we suggest scaling $\phi_t^\DR$ as $\xi_t := \phi_t^\DR / (k_t + 1)$ and setting $\lambda_t^L(\candpolval)$ as
  \begin{align}
    \lambda_t^L(\candpolval) &:= \sqrt{\frac{2 \log(1/\alpha)}{\widehat \sigma_{t-1}^2 t\log (1 + t)}} \land \frac{c}{k_t + \candpolval},~~~\text{where}\label{eq:betting-strategy-prpl}\\
    \widehat \sigma_t^2 &:= \frac{\sigma_0^2 + \sum_{i=1}^t (\xi_i - \widebar \xi_{i})^2}{t + 1},~~~\text{and}~~~\widebar \xi_t := \left ( \frac{1}{t}\sum_{i=1}^t \xi_i\right ) \land \frac{1}{k_t+1} \label{eq:betting-strategy-prpl-sigmahat2},
  \end{align}
  with a similar definition for $\lambda_t^U(\candpolval)$ but with $c/(k_t + \candpolval)$ replaced by $c/(k_t + (1-\candpolval))$. Here, $c \in (0, 1)$ is some truncation scale, reasonable values of which may lie between $1/4$ and $3/4$, but it is of relatively minor practical importance, and for sufficiently large $t$, the choice of $c$ will be inconsequential. A justification for why $\lambda_t^L(\candpolval)$ is a sensible choice can be found in \citet[Section B.1]{waudby2020estimating}, but the practitioner is nevertheless free to use any other sequence of bets, as long as they are predictable and satisfy the aforementioned boundedness constraints. Furthermore, as in \citet{waudby2020estimating}, when $\lambda_t^L(\candpolval)$ is chosen as above, the product in \eqref{eq:dr-lower} is quasiconvex in $\candpolval \in [0, 1]$ and hence the infima in \cref{theorem:dr-fixed-policy-value} (and \cref{proposition:kmr}) can be computed straightforwardly via line or grid search (see \citet[Section A.5]{waudby2020estimating}).
  
  The sequence of nonnegative $\infseqt{k_t}$ that truncate the reward predictors can also be chosen in any way as long as they are predictable. There are several heuristics that one might use, with increasing levels of complexity. One option is to have a prior guess for $\wmax$ (or some value $\wmax'$ that the practitioner believes will upper-bound most importance weights) and set $k_t = \wmax' / C$ for some $C \geq 1$, e.g.~$C = 2$. For a more adaptive option, one could set $k_t := \mathrm{median}(w_1, \dots, w_{t-1})$, or even try out a grid of values $\{k^{(1)}, \dots, k^{(J)}\}$ and choose the $k^{(j)}$ that would have yielded the tightest \cs{}s in hindsight. Nevertheless, all three of these options yield nonasymptotically valid $(1-\alpha)$-\cs{}s for $\polval$.
\end{remark}

\begin{remark}[Why truncate the reward predictor $\widehat r_t$?]\label{remark:DR-why-truncate}
  Readers familiar with doubly robust estimation in causal inference or contextual bandits will notice that if $k_1 = k_2 = \cdots = \infty$, then $\phi_t^\DRL$ takes the form of a classical doubly robust estimator of the policy value, and that such estimators often vastly outperform those based on importance weighting alone (and in many cases, are provably more efficient, at least in an asymptotic sense), so why would we want to truncate $\widehat r_t$ at all?

  The reason has to do with the fact that for nonasymptotic inference, we exploit the lower-boundedness of $\phi_t^\DRL$ in order to show that the product in \eqref{eq:dr-lower} is a nonnegative martingale. However, if we introduce a non-truncated reward predictor, we can only say that $\phi_t^\DRL$ is lower-bounded by $-\wmax$, which we do not want to assume knowledge of (or that it is finite at all). Truncation of $\widehat r_t$ allows us to occupy a middle ground, so that many of the efficiency gains from doubly robust estimation can be realized, without entirely losing the lower-boundedness structure of $\phi_t^\DRL$. 

  This same line of thought helps to illustrate why including a reward predictor tightens our \cs{}s for large $t$, potentially at the expense of tightness for smaller $t$. Notice that truncating the reward predictor at $k_t / w_t$ simply restricts the tuning parameter $\lambda_t^L(\candpolval)$ to lie in $[0, (\candpolval + k_t)^{-1})$ instead of $[0, \candpolval^{-1})$ for importance weighting. Without dwelling on the details too much, it is known that smaller values of $\lambda_t$ correspond to \cs{}s and \ci{}s being tighter for larger $t$ --- e.g. the role that $\lambda$ plays in Hoeffding's \ci{}s looks like $\sqrt{\log (2/\alpha)/ 2n}$ \citep{hoeffding_probability_1963} --- but we refer the reader to papers on \cs{}s for more in-depth discussions \citep{howard2018uniform,waudby2020estimating}. The important takeaway for our purposes here, is that larger $k_t$ corresponds to more variance adaptivity via double robustness, but does more to restrict the \cs{}s tightness at small $t$. Nevertheless, this tradeoff is clearly worth it in some cases (see \cref{fig:iw-vs-dr}).
  
  We note that the idea to truncate $\widehat r_t$ based on $k_t / w_t$ was inspired by the so-called ``reduced-variance'' estimators of \citet{zimmert2019connections,zimmert2021tsallis}. However, their reduced-variance estimators are slightly different since they multiply by an indicator $\1(w_i \leq \eta)$ for some $\eta \geq 0$, which sends the reward predictor to zero for large importance weights, whereas ours only truncates the reward predictor.
\end{remark}

\begin{remark}[Mirroring trick for upper \cs{}s]\label{remark:mirroring-trick}
Notice that in~\cref{proposition:kmr} and~\cref{theorem:dr-fixed-policy-value}, upper \cs{}s for $\polval$ were obtained by importance weighting $1-R_t$ rather than $R_t$ to obtain a \emph{lower} \cs{} for $1-\polval$, which was then translated into an \emph{upper} \cs{} for $\polval$. This ``mirroring trick'' --- first used in the OPE setting by \citet{thomas2015high} to the best of our knowledge --- applies to all of the results that follow, but for the sake of succinctness, we will only explicitly write the lower \cs{}s.
\end{remark}

\subsection{Closed-form confidence sequences}\label{section:prpl-cs}

In~\cref{theorem:dr-fixed-policy-value}, we derived \cs{}s for the policy value that generalize and improve on prior work \citep{karampatziakis2021off}. These bounds are empirically tight and can be computed efficiently, but are not closed-form, which may be desirable in practice. In this section, we derive a simple, closed-form, variance-adaptive \cs{} for the fixed policy value $\polval := \EE_{\pi}(R_t \mid \history_{t-1})$. Let
\begin{equation}
    \xi_t := \frac{\phi_t^\DRL}{k_t + 1},~~~ \text{and}~~~ \widehat \xi_{t-1} := \left ( \frac{1}{t-1} \sum_{i=1}^{t-1} \xi_i \right ) \land \frac{1}{k_t + 1}.
\end{equation}
With the above notation in mind, we are ready to state the main result of this section.

\begin{proposition}[Closed-form predictable plug-in \cs{} for $\polval$]\label{proposition:prpl-cs}
  Given contextual bandit data $\infseqt{X_t, A_t, R_t}$ with $[0, 1]$-valued rewards, choose nonnegative and predictable tuning parameters $\infseqt{k_t}$,  and define $\infseqt{\lambda_t}$ as
  \begin{equation}
    \label{eq:prplcs}
    \lambda_t := \sqrt{\frac{2 \log(1/\alpha)}{\widehat \sigma_{t-1}^2 t\log (1 + t)}} \land c,~~\widehat \sigma_t^2 := \frac{\sigma_0^2 + \sum_{i=1}^t (\xi_i - \widebar \xi_{i})^2}{t + 1},~~\widebar \xi_t := \left ( \frac{1}{t}\sum_{i=1}^t \xi_i\right ) \land \frac{1}{k_t+1},
  \end{equation}
  where $c \in (0, 1)$ is some truncation parameter (reasonable values of which may include 1/2 or 3/4); and $\xi_0 \in (0, 1)$ and $\sigma_0^2 > 0$  are some user-chosen parameters that can be thought of as prior guesses for the mean and variance of $\xi$, respectively. Then,
  \begin{equation}
    L_t^\PrPl := \left ( \frac{\sum_{i=1}^t \lambda_i \xi_i}{\sum_{i=1}^t \lambda_i / (k_i + 1)} - \frac{\log(1/\alpha) + \sum_{i=1}^t \left ( \xi_i - \widehat \xi_{i-1} \right )^2\psi_E(\lambda_i)}{\sum_{i=1}^t \lambda_i / (k_i + 1)} \right ) \lor 0
  \end{equation}
  forms a lower $(1-\alpha)$-\cs{} for $\polval$. An analogous upper \cs{} $\infseqt{U_t^\PrPl}$ follows by mirroring (\cref{remark:mirroring-trick}).
\end{proposition}

\begin{figure}[!htbp]
  \centering
  \includegraphics[width=\textwidth]{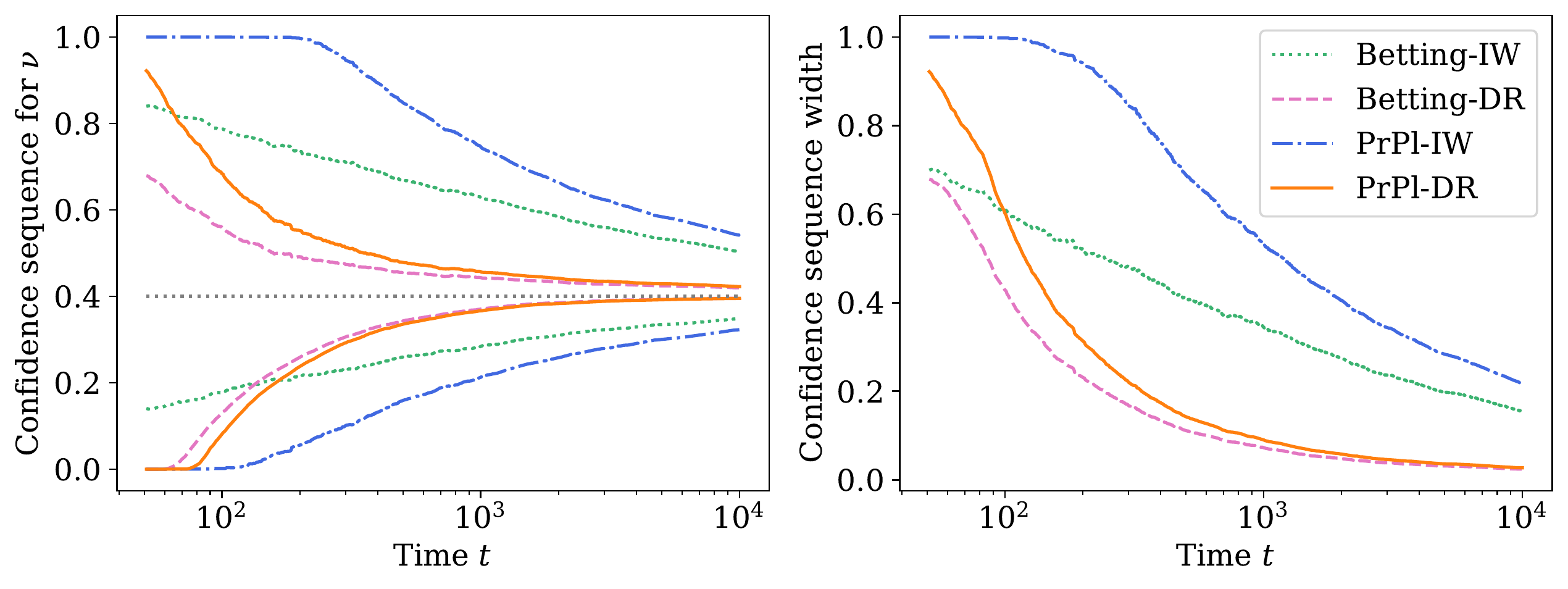}
  \caption{Betting-based (\cref{theorem:dr-fixed-policy-value}) and predictable plug-in (PrPl) (\cref{proposition:prpl-cs}) \cs{}s for $\polval$ with both importance-weighted (IW) and doubly robust (DR) variants. Notice that for both IW and DR \cs{}s, the betting-based approach of \cref{theorem:dr-fixed-policy-value} outperforms the PrPl \cs{}s. Nevertheless, the closed-form PrPl \cs{}s are simpler to implement, and can still benefit from doubly robust variance adaptation.}
  \label{fig:prpl-vs-betting}
\end{figure}
The proof can be found in \cref{proof:prpl-cs} and relies on Ville's inequality \citep{ville1939etude} applied to a predictable plug-in empirical Bernstein supermartingale similar to that of \citet[Section 3.2]{waudby2020estimating} but with a variant of Fan's inequality \citep{fan2015exponential} for lower-bounded random variables with upper-bounded means. As seen in \cref{fig:prpl-vs-betting}, the betting \cs{}s of \cref{theorem:dr-fixed-policy-value} still have better empirical performance than the closed-form predictable plug-in \cs{}s of \cref{proposition:prpl-cs}, but the latter are more computationally and analytically convenient, and are valid under the same set of assumptions. A similar phenomenon was observed by \citet{waudby2020estimating} for bounded random variables (outside the context of OPE). In the on-policy setting with all importance weights set to 1 and without a reward predictor, \cref{proposition:prpl-cs} recovers \citet[Theorem 2]{waudby2020estimating}.

It is important not to confuse $\widehat \xi_t$ with $\widebar \xi_t$. The difference between them may seem minor since the former simply has access to one less data point than the latter, but they play two very different roles in \cref{proposition:prpl-cs}: $\widehat \xi_{t-1}$ is a \emph{predictable} sample mean that shows up in the width of $L_t^\PrPl$ explicitly, and its predictability is fundamental to the proof technique, while $\widebar \xi_t$ is just used as a tool to obtain better estimates of $\widehat \sigma_t^2$ so that they can be plugged in to the tuning parameters $\lambda_t$. Consequently, $\widehat \xi_t$ can be found in \cs{}s that rely on a similar proof technique (such as \cref{theorem:conjmix-eb}), while $\widebar \xi_t$ can be found in other \cs{}s that make use of predictable tuning parameters (such as \cref{theorem:dr-fixed-policy-value}).

\subsection{Fixed-time confidence intervals}
While this paper is focused on time-uniform \cs{}s for OPE, our methods also naturally give rise to fixed-time \ci{}s that are \emph{not} anytime-valid but can still benefit from our general techniques. In this section, we will briefly discuss what minor modifications are needed to derive sharp fixed-time instantiations of our otherwise time-uniform bounds. We will also compare our fixed-time \ci{}s to the \ci{}s of \citet{thomas2015high}, but this comparison is by no means comprehensive. Indeed, our goal is not to show that our methods are ``better'' than prior work, even if some simulations may suggest this --- instead, we aim to provide the reader with some context as to how our fixed-time instantiations fit within the broader literature on \ci{}s for OPE\@.

\paragraph{Confidence intervals for policy values.} We begin by deriving fixed-time analogues of the \cs{}s for $\polval$ presented in \cref{theorem:dr-fixed-policy-value} and \cref{proposition:prpl-cs} --- the former being a ``betting-style'' \cs{} that is very tight in practice, and the latter being a closed-form predictable plug-in (PrPl) \cs{} that is slightly more analytically and computationally convenient. In both cases, our suggested modification is essentially the same: choose a predictable sequence $\seq{\lambda_t}{t}{1}{n}$ that is tuned for the desired sample size $n$ --- to be elaborated on shortly --- and take the intersection of the implicit \cs{} $(C_t)_{t=1}^n$ that is formed from times 1 through $n$. Concretely, define the predictable sequence $(\lambda_t)_{t=1}^n$ given by
\begin{equation}
  \dot \lambda_{t,n} := \sqrt{\frac{2\log(1/\alpha)}{n\widehat \sigma_{t-1}^2}},
\end{equation}
where $\widehat \sigma_t^2$ is given as in \eqref{eq:betting-strategy-prpl-sigmahat2}. Then, we have the following corollary for betting-style \ci{}s for $\polval$.
\begin{corollary}\label{corollary:betting-ci}
  Let $(X_t, A_t, R_t)_{t=1}^n$ be a finite sequence of contextual bandit data with $[0, 1]$-valued rewards and define $(\lambda_t^L(\candpolval))_{t=1}^n$ and $(\lambda_t^U(\candpolval))_{t=1}^n$ by
  \begin{equation}\label{eq:betting-ci-lambdas}
    \lambda_t^L(\candpolval) := \dot \lambda_{t,n} \land \frac{c}{k_t + \candpolval}, ~~~\text{and}~~~\lambda_t^U(\candpolval) := \dot \lambda_{t,n} \land \frac{c}{k_t + (1- \candpolval)},
  \end{equation}
  where $c \in (0, 1)$ is some truncation scale as in~\cref{theorem:dr-fixed-policy-value}.
  Let $L_t^\DR$ and $U_t^\DR$ be as in~\eqref{eq:dr-lower} and~\eqref{eq:dr-upper}. Then,
  \begin{equation}\label{eq:betting-ci}
    \dot L_n := \max_{1 \leq t \leq n} L_t^\DR ~~~\text{and}~~~ \dot U_n := \min_{1 \leq t \leq n} U_t^\DR
  \end{equation}
  form lower and upper $(1-\alpha)$-\ci{}s for $\polval$, respectively, meaning $\PP(\nu \in [\dot L_n, \dot U_n]) \geq 1-\alpha$.
\end{corollary}
\begin{figure}[!htbp]
  \centering
  \includegraphics[width=\textwidth]{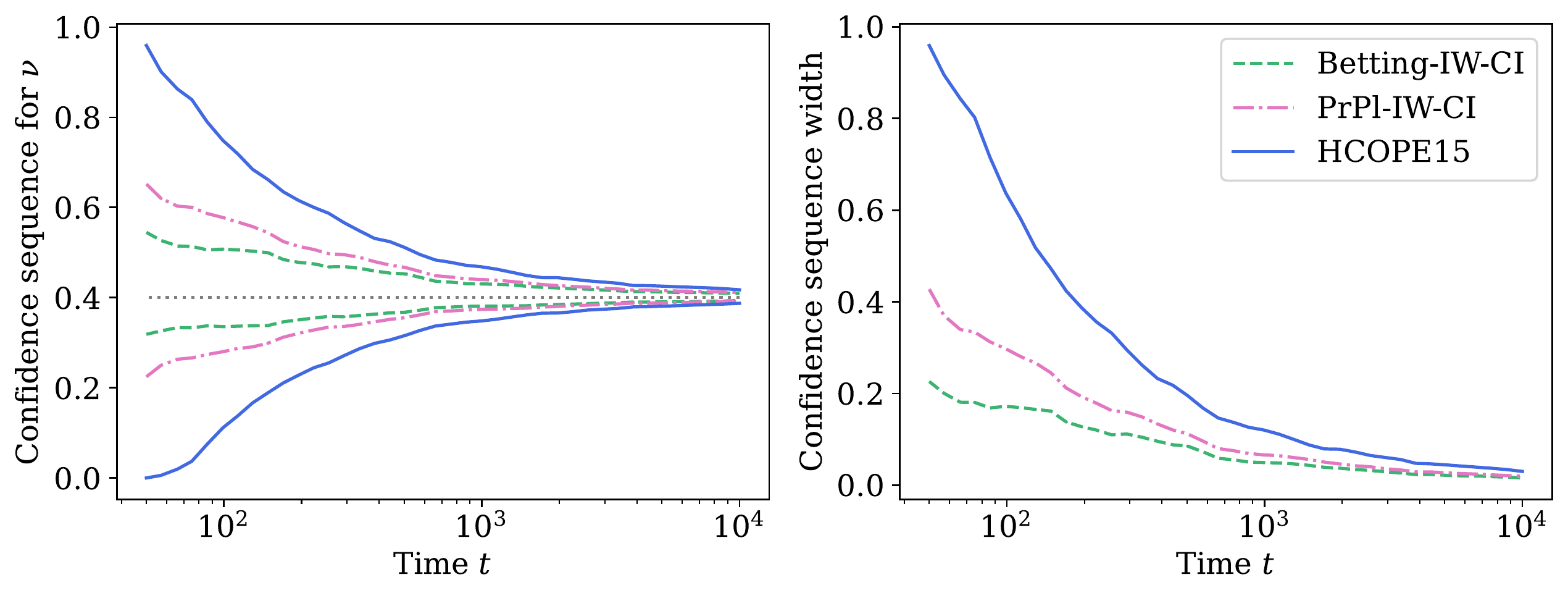}
  \caption{Fixed-time 90\% confidence intervals for $\polval$ using three different methods: a betting-based \ci{} (\cref{corollary:betting-ci}), a predictable plug-in (PrPl) \ci{} (\cref{corollary:prpl-ci}), and those presented in a paper entitled ``High-confidence off-policy evaluation'' (HCOPE15) by \citet{thomas2015high}. Notice that the betting-based \ci{}  outperforms the closed-form PrPl \ci{}, which itself significantly outperforms the bounds in \citet{thomas2015high}.}
  \label{fig:fixed-time-cis}
\end{figure}
\cref{corollary:betting-ci} is an immediate consequence of~\cref{theorem:dr-fixed-policy-value} where the sequence of tuning parameters was chosen to tighten the \ci{} for the sample size $n$. This particular choice of $\seq{\lambda_t^L(\candpolval)}{t}{1}{n}$ is inspired by the fact that the product in \eqref{eq:dr-lower} resembles an exponential supermartingale whose resulting \ci{} can be tightened using tuning parameters that are well-estimated by \eqref{eq:betting-ci-lambdas}. For more details, we refer the reader to \citet[Section 3]{waudby2020estimating}. The fact that the maximum and minimum can be taken in~\eqref{eq:betting-ci} follows from the fact that $(L_t^\DR)_{t=1}^n$ satisfies $\PP \left (\forall t \in \{1, \dots, n\},\ \polval \geq L_t^\DR \right ) \geq 1-\alpha$, and similarly for the upper \ci{} $\dot U_n$. \cref{fig:fixed-time-cis} demonstrates what these \ci{}s may look like in practice.

Similar to how \cref{corollary:betting-ci} is a fixed-time instantiation of \cref{theorem:dr-fixed-policy-value}, the following corollary is a fixed-time instantiation of the closed-form PrPl \cs{}s of \cref{proposition:prpl-cs}.
\begin{corollary}\label{corollary:prpl-ci}
  Let $(X_t, A_t, R_t)_{t=1}^n$ be a finite sequence of contextual bandit data with $[0, 1]$-valued rewards and define $(\lambda_t)_{t=1}^n$ by
  \begin{equation}
    \lambda_t := \dot \lambda_{t,n} \land c,
  \end{equation}
  with $c$ chosen as in \cref{proposition:prpl-cs}. Then, with $L_t^\PrPl$ and $U_t^\PrPl$ defined in \cref{proposition:prpl-cs}, we have that
  \begin{equation}
    \dot L_n^\PrPl := \max_{1 \leq t \leq n} L_t^\PrPl ~~~\text{and}~~~ \dot U_n^\PrPl := \min_{1 \leq t \leq n} U_t^\PrPl
  \end{equation}
  form lower and upper $(1-\alpha)$-\ci{}s for $\polval$, respectively, meaning $\PP( \nu \in [\dot L_n^\PrPl, \dot U_n^\PrPl] ) \geq 1-\alpha$.
\end{corollary}
\cref{corollary:prpl-ci} is an immediate consequence of \cref{proposition:prpl-cs} instantiated for a different choice of $\lambda_t$ and with an intersection being taken over the implicit \cs{} from times 1 through $n$.

While the methods of this section improve on past work both theoretically and empirically, each of our results thus far have assumed that $\polval \equiv \EE_{\pi}(R_t)$ is fixed and does not change over time, an assumption that we may not always wish to make in practice (e.g.~if the environment is nonstationary). Fortunately, it is still possible to design \cs{}s that capture an interpretable parameter: the \emph{time-varying average policy value thus far}. However, we will need completely different supermartingales to achieve this, which we outline in the following section.

\section{Inference for time-varying policy values}\label{section:time-varying}
Let us now consider the more challenging task of performing anytime-valid off-policy inference for a time-varying average policy value. Concretely, suppose that the value of the $[0, 1]$-bounded reward $R_t$ under policy $\pi$ is given by $\polval_t := \EE_{\pi}(R_t \mid \history_{t-1}) \in [0, 1]$, and hence $(\polval_t)_{t=1}^\infty$ is now a \emph{sequence} of conditional policy values. Our goal is to derive \cs{}s for $(\polvalt)_{t=1}^\infty$ where $\polvalt := \frac{1}{t} \sum_{i=1}^t \polval_i$ is the \emph{average conditional policy value so far}. In addition to satisfying desiderata 1--5 in \cref{section:desiderata} our \cs{}s will impose no restrictions on how $(\polval_t)_{t=1}^\infty$ changes over time. Unfortunately, the techniques of \cref{proposition:kmr} and \cref{theorem:dr-fixed-policy-value} will not work here because their underlying test statistics
cannot be written explicitly as functions of a candidate average policy value $\candpolvalt := \frac{1}{t}\sum_{i=1}^t \polval_i$, but only of the entire candidate tuple $(\polval_1 \dots, \polval_t)$. To remedy this, we will rely on test statistics that are functions of a candidate average $\candpolvalt$ to derive \cs{}s for $\polvalt$ in Theorems~\ref{theorem:conjmix-eb} and~\ref{proposition:lil-eb}. An empirical demonstration of the failure of \cref{theorem:dr-fixed-policy-value} juxtaposed with the remedy provided by \cref{theorem:conjmix-eb} can be seen in the left-hand side of \cref{fig:eb_vs_others} and a comparison between \cref{theorem:conjmix-eb} and \cref{proposition:lil-eb} can be seen in the right-hand side of the same figure.

We will present two \cs{}s for $\polvalt$: (1) the ``empirical Bernstein'' \cs{} in \cref{theorem:conjmix-eb} whose underlying supermartingale is constructed using Robbins' method of mixtures \citep{robbins1970statistical}, and (2) the ``iterated logarithm'' \cs{} in \cref{proposition:lil-eb} which uses the stitching technique. Both yield time-uniform, nonasymptotically valid bounds, and are easy to compute. However, the former tends to have better empirical performance in finite samples, while the latter achieves the (optimal) rate of convergence, matching the law of the iterated logarithm. Nevertheless, both boundaries shrink to zero-width at a rate of $\widetilde O(\sqrt{V_t} / t)$ \citep{howard2018uniform}. Here, $\widetilde O(\cdot)$ means $O(\cdot)$ up to logarithmic factors.

In order to write down the empirical Bernstein \cs{}, we first define the scaled doubly robust pseudo-outcomes $\xi_t := \phi_t / (1+k)$ where $\phi_t \equiv \phi_t^\DRL$ is given in \eqref{eq:dr-outcomes-lower} and with all $k_t$ equal to a fixed $k$. Then, define the corresponding centered sum process $\infseqt{S_t(\candpolvalt)}$ and variance process $\infseqt{V_t}$ given by
\begin{align}
  &S_t(\candpolvalt) := \sum_{i=1}^t \xi_i - \frac{t\candpolvalt}{1 + k}, ~~\text{and} \label{eq:sum-process}\\
  &V_t := \sum_{i=1}^t (\xi_i - \widehat \xi_{i-1})^2, ~~\text{where } \widehat \xi_{t} :=  \left ( \frac{1}{t} \sum_{i=1}^{t} \xi_i \right )\land \frac{1}{1 + k}, \label{eq:variance-process}
\end{align}
and $\xi_0 \in [0, (1+k)^{-1}]$ is chosen by the user. With this setup and notation in mind, we are ready to state the empirical Bernstein \cs{} for $\polvalt$.

\begin{theorem}[Empirical Bernstein confidence sequence for $\polvalt$]\label{theorem:conjmix-eb}
  Let $\infseqt{X_t, A_t, R_t}$ be an infinite sequence of contextual bandit data with $[0, 1]$-valued rewards generated by the sequence of policies $\infseqt{h_t}$, and let $\infseqt{S_t(\candpolvalt)}$ and $\infseqt{V_t}$ be the centered sum and variance processes as in~\eqref{eq:sum-process} and~\eqref{eq:variance-process}. Then for any $\rho > 0$,
  \begin{equation}\label{eq:gamma-exponential-nsm}
M_t^\EB(\polvalt) := \left(\frac{\rho^\rho e^{-\rho}}{\Gamma(\rho) - \Gamma(\rho, \rho)}\right) \left(\frac{1}{V_t + \rho}\right) \onefone(1, V_t + \rho + 1, S_t(\polvalt) + V_t + \rho)
  \end{equation}
  forms a nonnegative supermartingale starting at one, where $\onefone(\cdot, \cdot ,\cdot )$ is Kummer's confluent hypergeometric function and $\Gamma(\cdot, \cdot)$ is the upper incomplete gamma function. Consequently,
  \begin{equation}
    L_t^\EB := \inf \left \{ \candpolvalt \in [0, 1] : M_t^\EB(\candpolvalt) < 1/\alpha \right \}
  \end{equation}
  forms a lower $(1-\alpha)$-\cs{} for $\polvalt$, meaning $\PP(\forall t \in \NN,\ \polvalt \geq L_t^\EB) \geq 1-\alpha$. Similarly, an upper \cs{} can be derived by using $\phi_t^\DRU$ defined in \eqref{eq:dr-outcomes-upper} and employing the mirroring trick described in \cref{remark:mirroring-trick}.
\end{theorem}
The proof of~\cref{theorem:conjmix-eb} can be found in~\cref{proof:conjmix-eb} and relies on an inequality due to~\citet{fan2015exponential} along with a mixture supermartingale analogous to that of~\citet[Proposition 9]{howard2018uniform}. In the on-policy setting with all importance weights set to 1 and with no reward predictor (and hence $k = 0$), we have that \cref{theorem:conjmix-eb} recovers the gamma-exponential mixture \cs{} of \citet[Proposition 9]{howard2018uniform}.

The tuning parameter $\rho > 0$ effectively dictates the neighborhood of intrinsic time --- i.e.~the value of $V_t$ --- at which $L_t^\EB$ is tightest, and it is rather straightforward to choose $\rho > 0$ given this interpretation. Following \citet{howard2018uniform}, $\rho >0$ can be chosen to (approximately) tighten $L_t^\EB$ at $V_t = V^\star$ by setting
\begin{equation}
    \rho(V^\star) := \sqrt{\frac{-2\log \alpha + \log(-2\log\alpha + 1) }{V^\star}}.
\end{equation}
Nevertheless, $L_t^\EB$ forms a valid lower $(1-\alpha)$-\cs{} for $\polvalt$ regardless of how $\rho > 0$ is chosen, as long as this is done data-independently.

Readers familiar with gamma-exponential mixture supermartingales such as those in \citet{howard_exponential_2018,howard2018uniform} and \citet{choe2021comparing} may have expected to see a lower incomplete gamma function $\gamma(a,b)$ instead of $\onefone{(1, a,b)}$. Indeed, $\onefone(1,a,b)$ reduces to $\gamma(a,b)$ when $b \equiv S_t(\polvalt) + V_t + \rho$ is nonnegative, but unlike the lower incomplete gamma, $\onefone{}(1, a,b)$ is well-defined when $b < 0$. Writing \eqref{eq:gamma-exponential-nsm} in terms of a lower incomplete gamma would have required lower-bounding this term by a piece-wise function when $b \equiv S_t + V_t + \rho < 0$ as in \citet[Appendix C]{choe2021comparing}; see \cref{section:onefone-vs-lowerincgamma} for  details.

While $L_t^\EB$ is not a closed-form bound, it can be computed efficiently using root-finding algorithms, and it can also be shown to achieve an asymptotic width of $O(\sqrt{V_t \log V_t} / t)$ (the justifications in \citet{howard2018uniform} carry over to this scenario). While this rate is sufficient for deriving \cs{}s with strong empirical performance that shrink to zero-width, one can derive \cs{}s that achieve an improved rate of $O(\sqrt{V_t \log \log V_t} / t)$ using a different technique known as ``stitching''.

\begin{figure}[!htbp]
  \centering
  \includegraphics[width=0.49\textwidth]{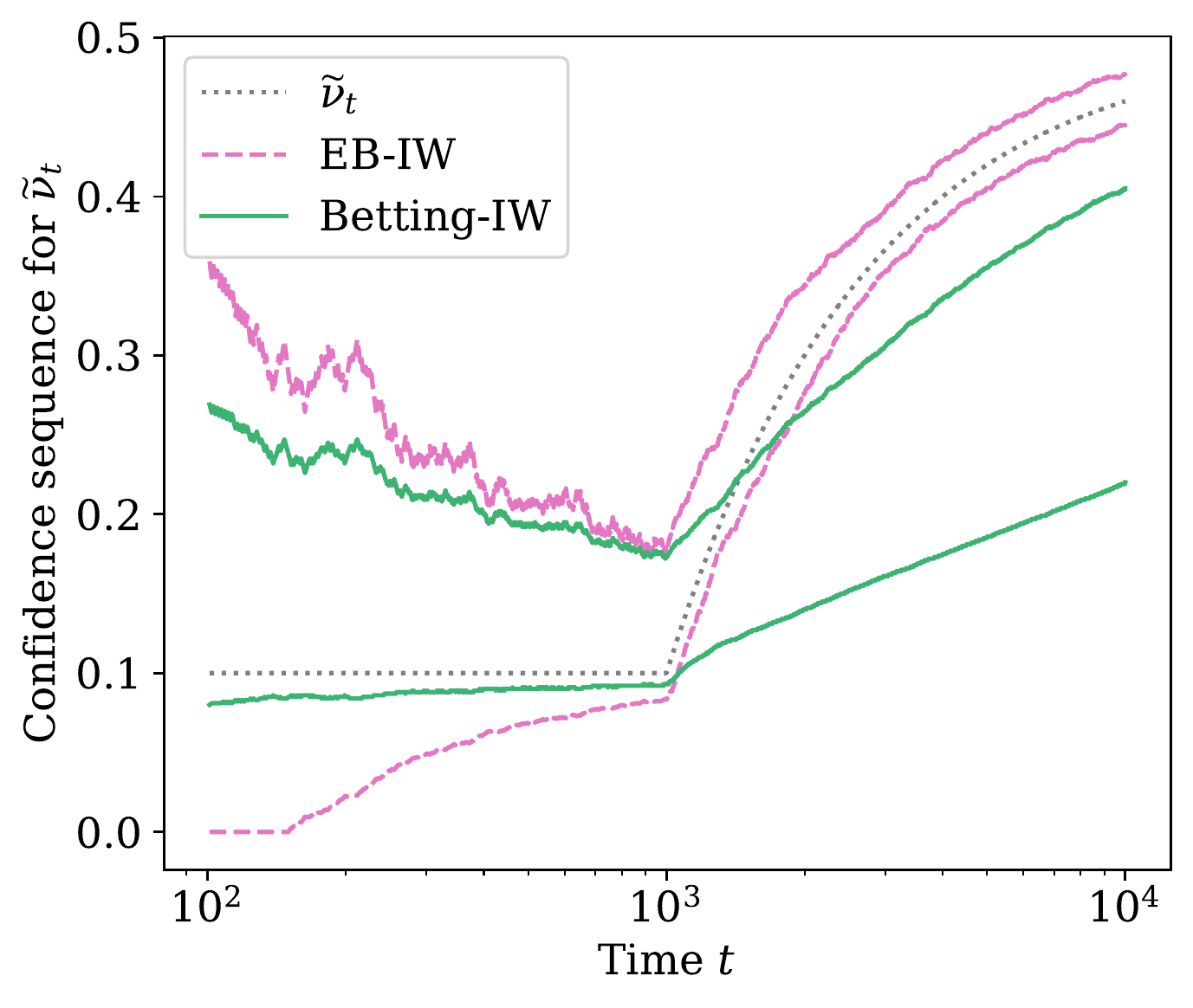}
  \includegraphics[width=0.49\textwidth]{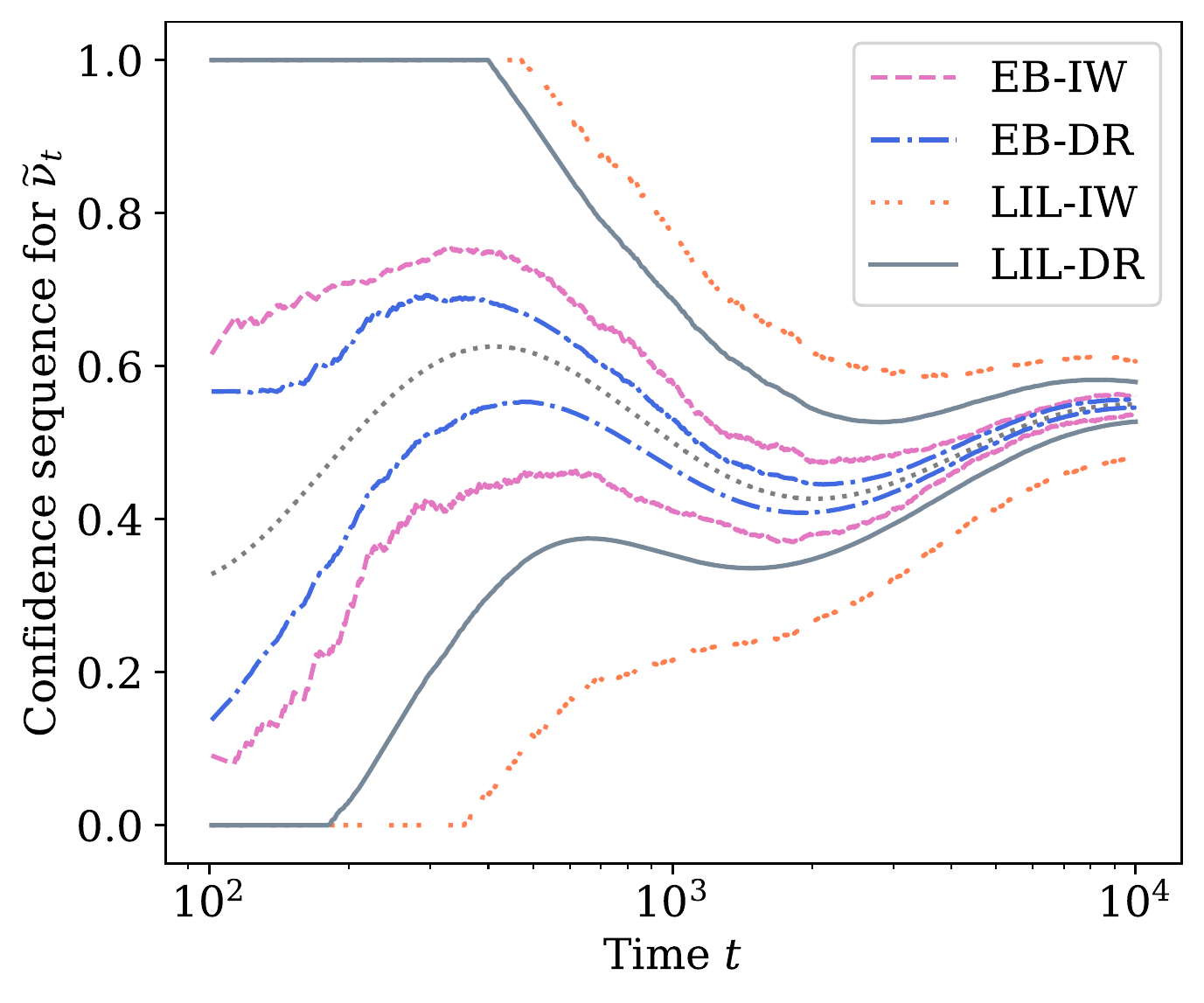}
  \caption{Various \cs{}s for the time-varying policy value $\polvalt$. The left-hand side plot illustrates that while the betting-style \cs{} of \cref{theorem:dr-fixed-policy-value} is tight when $\polvalt$ remains fixed, it fails to cover when $\polvalt$ changes (in this case, there is an abrupt change at $t = 1000$). The right-hand side plot illustrates how \cref{theorem:conjmix-eb} and \cref{proposition:lil-eb} compare, both using their importance-weighted (IW) and doubly robust (DR) variants. Notice that while LIL-IW and LIL-DR attain optimal rates of convergence, the empirical Bernstein \cs{}s (EB-IW and EB-DR) are much tighter in practice. In both cases, the DR variant outperforms the IW variant due to the reward being easy to predict in this particular example.}
  \label{fig:eb_vs_others}
\end{figure}

\begin{proposition}[Variance-adaptive iterated logarithm confidence sequence for $\polvalt$]\label{proposition:lil-eb}
  Let $\widebar V_t := V_t \lor 1$ where $V_t$ is given as in~\eqref{eq:variance-process}.
  Define the function $\ell_t(\alpha)$
  \begin{align}
    \ell_t(\alpha) &:= 2 \log \left ( \log \widebar V_t + 1 \right ) + \log \left ( \frac{1.65}{\alpha} \right ).
  \end{align}
  Then we have that under the same conditions as \cref{theorem:conjmix-eb},
  \begin{equation}\label{eq:lil-lower}
    L_t^\LIL := (k + 1) \left ( \frac{1}{t}\sum_{i=1}^t \xi_i - \frac{\sqrt{2.13 \ell_t(\alpha) \widebar V_t + 1.76 \ell_t(\alpha)^2}}{t} - \frac{1.33 \ell_t(\alpha)^2}{t} \right ) \lor 0 \\
  \end{equation}
  forms a lower $(1-\alpha)$-\cs{} for $\polvalt := \frac{1}{t} \sum_{i=1}^t \polval_i$, meaning $\PP(\forall t,\ \polvalt \geq L_t^\LIL) \geq 1-\alpha$. An analogous upper \cs{} $\infseqt{U_t^\LIL}$ can be derived using the mirroring trick of~\cref{remark:mirroring-trick}.
\end{proposition}
The proof in \cref{proof:lil-eb} uses the ``stitching'' (sometimes called ``peeling'') technique that is common in the derivation of LIL-type bounds \citep{darling1967confidence,jamieson2014lil,kaufmann2016complexity} applied to linear sub-exponential boundaries. Importantly, $L_t^\LIL \asymp \Ocal\left (t^{-1}\sqrt{ V_t \log \log V_t} \right )$ which matches the unimprovable rate implied by the law of the iterated logarithm. We take a maximum with 0 in $L_t^\LIL$ (and hence an implicit minimum with 1 in $U_t^\LIL$) since $\polvalt \in [0, 1]$  for every $t \in \NN$ by assumption. Of course, if one is in a stationary environment so that $\polval_1 = \polval_2 = \cdots = \polval$, then $\polvalt = \polval$, and hence $(L_t^\LIL, U_t^\LIL)$ forms a $(1-\alpha)$-CS for $\polval$.

\paragraph{Comparison of Theorems~\ref{theorem:dr-fixed-policy-value} and~\ref{theorem:conjmix-eb} with prior work.}
To the best of our knowledge,~\citet{thomas2015high} were the first to derive nonasymptotic confidence intervals for policy values in contextual bandits, and they did so without knowledge of $\wmax$. However, their bounds are not time-uniform, and the authors do not consider time-varying policy values nor data-dependent logging policies $\infseqt{h_t}$. Four other prior works stand out as being related to the results of this section, namely~\citet{karampatziakis2021off},~\citet[Section 4.2]{howard2018uniform},~\citet{bibaut2021post}, and~\citet{zhan2021off} and we discuss each of them in some detail below. Note that in the last row labeled ``Doubly robust'' of \cref{table:time-varying}, we are referring to the property of confidence sets to be potentially sharpened in the presence of regression estimators without compromising validity (as discussed in the paragraphs surrounding \cref{theorem:dr-fixed-policy-value}).

\begin{itemize}
\item \textbf{KMR21:}.
  The off-policy \cs{}s of~\citet{karampatziakis2021off} --- reviewed in~\cref{proposition:kmr} --- can in several ways be seen as an improvement of~\citet{thomas2015high} since they are time-uniform, in addition to being empirically tight. As discussed in~\cref{section:warmup}, their importance-weighted off-policy \cs{}s are strictly generalized and extended in our~\cref{theorem:dr-fixed-policy-value}, but their result nevertheless yields the current state-of-the-art \cs{}s for $\polval$.

\item \textbf{BDKCvdL21 \& ZHHA21:} \citet{bibaut2021post} and~\citet{zhan2021off} both study the off-policy inference problem in contextual bandits from an asymptotic point of view. Their off-policy estimators take the form of sample averages of influence functions --- what \citet{bibaut2021post} refer to as the canonical gradient --- to which martingale central limit theorems may be applied to obtain asymptotically valid inference.

In contrast to our work, the confidence intervals of \citet{bibaut2021post} and \citet{zhan2021off} (a) are asymptotic, and hence do not have finite-sample guarantees, (b) are not time-uniform, and hence cannot be used at stopping times, and (c) do not track time-varying policy values.


  \item \textbf{HRMS21:} \citet[Section 4.2]{howard2018uniform} derive time-uniform, nonasymptotic \cs{}s for the average treatment effect (ATE) in randomized experiments. The main difference between our results and those of~\cite[Section 4.2]{howard2018uniform} is that they do not study the contextual bandit off-policy evaluation problem. However, since estimating the ATE in randomized experiments can be seen as a special case of the contextual bandit problem, it is natural to wonder how our approach differs in this special case. The main difference here is that~\citet{howard2018uniform} require knowledge of $\wmax$ --- or equivalently in their setup, the maximal and minimal propensity scores --- while ours do not. Moreover, our results allow $\wmax$ to be infinite and nevertheless enjoy variance-adaptivity. See~\cref{section:ate-randomized-expts} for a more detailed discussion of the implications of our bounds for ATE estimation in randomized experiments.
\end{itemize}


\begin{table}[!htbp]
  \caption{Comparison of various \cs{}s and \ci{}s for mean off-policy values.}
  \label{table:time-varying}
  \centering
  \begin{tabular}{|l|c|c|c|c|c|}
    \hline
                            & KMR21      & BDKCvdL21 \& ZHHA21 & HRMS21     & Thm.~\ref{theorem:dr-fixed-policy-value} & Thm.~\ref{theorem:conjmix-eb} \\
    \hline
    Contextual bandits            & \checkmark & \checkmark &            & \checkmark                               & \checkmark                    \\
    Time-varying rewards            &            &            & \checkmark &                                          & \checkmark                    \\
    Nonasymptotic           & \checkmark &            & \checkmark & \checkmark                               & \checkmark                    \\
    Time-uniform            & \checkmark &            & \checkmark & \checkmark                               & \checkmark                    \\
    Predictable $\infseqt{h_t}$ &            & \checkmark & \checkmark & \checkmark                               & \checkmark                    \\
    $\wmax$-free            & \checkmark & \checkmark &            & \checkmark                               & \checkmark                    \\
    Doubly robust         &            & \checkmark & \checkmark & \checkmark                               & \checkmark                    \\
    \hline
  \end{tabular}
\end{table}

\subsection{A remark on policy value differences}\label{section:policy-value-differences}
The results in this paper have taken the form of \cs{}s for policy values, e.g.~a sequence of sets $[L_t, U_t]_{t=1}^\infty$ such that $\PP(\forall t \in \NN,\ \polval \in [L_t, U_t]) \geq 1-\alpha$, but it may be of interest to directly estimate policy value \emph{differences} --- e.g.~$\poldiff \equiv \polval_1 - \polval_2 \equiv \polval(\pi_1) - \polval(\pi_2)$, where $\polval(\pi)$ is the value of some policy $\pi$, and $\pi_1$ and $\pi_2$ are two policies we would like to compare. 
In many cases including the gated deployment problem studied by~\citet{karampatziakis2021off}, $\pi_1$ is some target policy of interest and $\pi_2 = h_1 = h_2 = \cdots = h$ is the logging policy so that $\polval(\pi_1) - \polval(h)$ can be interpreted as the additional value (or ``lift'') in the target policy $\pi_1$ over the logging policy. However, our setup allows for $\pi_1$ and $\pi_2$ to be any two policies that are absolutely continuous with respect to the logging policies.

Of course, one can always solve this problem by union bounding: construct $(1-\alpha/2)$-\cs{}s for $\polval_1$ and $\polval_2$ separately to yield a $(1-\alpha)$-\cs{} for their difference. However, it is possible to remove this small amount of slack introduced by union bounding and instead derive a \cs{} for the difference directly.

The idea is simple: rather than only leverage lower-boundedness of importance-weighted rewards $w_tR_t$, we construct a new random variable $\theta_t := w_t^{(1)}R_t - [1 - w_t^{(2)}(1-R_t)]$ and leverage its lower-boundedness directly --- here, $w_t^{(1)} = \pi_1 / h_t$ and $w_t^{(2)} = \pi_2 / h_t$ are the importance weights for policies $\pi_1$ and $\pi_2$, respectively. In particular, notice that
\begin{align}
  \EE \left [ \theta_t \right ] = \poldiff, ~~~\text{and}~~~ \theta_t &\geq -1, ~~\text{and hence}\\
  \frac{1}{2} \left [ \theta_t - \poldiff  \right ] &\geq - 1~~\text{almost surely.}\label{eq:centered-policy-diff}
\end{align}
Consequently,~\eqref{eq:centered-policy-diff} can be used in the proofs of our theorems to derive a \cs{} for $\poldiff$ directly, since those proofs fundamentally rely on the centered (i.e.~with their mean subtracted) random variables being almost surely lower-bounded by $-1$. For instance, we have that for any $(0, 1)$-valued predictable sequence $\infseqt{\lambda_t(\poldiff)}$, 
\begin{equation}
  M_t(\poldiff) := \prod_{i=1}^t \left ( 1 + \lambda_i(\poldiff) \cdot (\theta_t - \poldiff)/2 \right )
\end{equation}
forms a test martingale and hence $L_t := \inf \left \{ \poldiff' \in [-1, 1] : M_t(\poldiff') < 1/\alpha \right \}$ forms a lower $(1-\alpha)$-\cs{} for $\poldiff$.
As usual, the mirroring trick can be used to obtain an upper \cs{} for this policy value difference. Moreover, the above discussion can be extended to time-varying policy value differences and doubly robust pseudo-outcomes (rather than just their importance-weighted counterparts), as well as sequences of policies --- i.e.~analyzing the sequences $\infseqt{\pi_1^{(t)}}$ and $\infseqt{\pi_2^{(t)}}$ --- but we omit these derivations for the sake of brevity.

\subsection{Time-varying treatment effects in adaptive experiments}\label{section:ate-randomized-expts}

While this paper is focused on anytime-valid \emph{contextual bandit} inference --- i.e.~inference for policy values or their CDFs from contextual bandit data --- one can nevertheless view off-policy evaluation as a generalization of treatment effect estimation from adaptive experiments. Consequently, every single result in this paper also has powerful implications for nonasymptotic inference for treatment effects from such experiments. In this section, we will focus on adaptive experiments with binary treatments for simplicity, but the analogy extends to more general settings.

\paragraph{From contextual bandits to adaptive experiments with binary treatments.} The contextual bandit problem can be seen as a generalization of adaptive experiments since the latter has three key notational differences. 
\begin{enumerate}
\item The ``context'' $X_t$ is typically referred to as a ``covariate'' or a ``feature'', and may be used to represent baseline demographics and medical history in a clinical trial, for example.
\item The ``action'' $A_t$ (which is binary in this case) is referred to as a ``treatment'', and the policy $h_t$ is called the ``propensity score'', and is simply the probability of a subject with covariates $X_t$ receiving treatment $A_t = 1$ at time $t$.
\item The ``reward'' $R_t$ is often referred to as the ``outcome'' for subject $t$.
\end{enumerate}
There are many reasons why one may wish to run an adaptive sequential experiment rather than a simple Bernoulli($h$) experiment with a constant pre-specified propensity score $h$. Two simple examples include: (a) balancing designs such as Efron's biased coin~\citep{efron1971forcing} which vary the propensity scores $\infseqt{h_t}$ over time to ensure that treatment groups are ``balanced'' within certain levels of the covariates, and (b) the experimental designs of \citet{kato2020adaptive} which adaptively choose propensity scores to minimize the variance of the resulting doubly robust and inverse propensity-weighted (IPW) estimators, yielding sharper confidence sets. (In the language of contextual bandits and off-policy evaluation, IPW and importance weighting are equivalent.) Both (a) and (b) --- or any other design that varies propensity scores adaptively over time --- can be paired with the \cs{}s of the current paper.

\paragraph{Implications for causal inference in adaptive experiments.} From the perspective of treatment effect estimation, the current paper provides nonasymptotic, nonparametric, time-uniform inference for treatment effects, all without knowledge of the minimal propensity score $h_\mathrm{min} := \essinf_{t,a,x}h_t(a \mid x)$ and this essential infimum can even be 0 (as long as it is not attained, meaning each $h_t(a \mid x)$ is itself positive). Contrast this with prior work on nonasymptotic, nonparametric, time-uniform inference for treatment effects such as \citet[Section 4.2]{howard2018uniform}, which require \emph{a priori} knowledge of $h_\mathrm{min} > 0$. Their bounds necessarily scale with an implied upper-bound on the variance of $h_t^{-1}R_t$ implied by $h_\mathrm{min}^{-1}$ while ours only scale with the \emph{empirical variance} of $(h_t^{-1} R_t)_{t=1}^\infty$ --- the latter always being smaller.

Concretely,~\cref{theorem:dr-fixed-policy-value} can be used to derive \cs{}s for the average treatment effect from adaptively collected data in experiments with binary treatments and bounded outcomes. \cref{theorem:conjmix-eb} goes further, enabling the construction of \cs{}s for \emph{time-varying average treatment effects} similar to \citet[Section 4.2]{howard2018uniform}.
Finally, \cref{theorem:cdf} --- to be presented in \cref{section:cdf} --- allows for the construction of time-uniform confidence bands for the CDF of the outcome distribution under a given treatment.
Moreover, all of this is possible in a nonparametric, nonasymptotic framework, without knowledge (or strict positivity) of $h_\mathrm{min}$.
To the best of our knowledge, all three of these implied results are new in the literature for treatment effect estimation.

\subsection{Sequential testing and anytime \texorpdfstring{$p$}{p}-values for off-policy inference}\label{section:testing}

While we have thus far taken an \emph{estimation} perspective (i.e.~deriving \cs{}s and \ci{}s rather than $p$-values), all of our results have hypothesis testing analogues. In particular, the \cs{}s and \ci{}s developed in this paper have all been built by first deriving implicit \emph{$e$-processes}.
Formally, given a set of distributions $\Pcal_0$ (referred to as ``the null hypothesis''), an $e$-process $E \equiv \infseqt{E_t}$ for $\Pcal_0$ is a nonnegative process such that $\EE_P[E_\tau] \leq 1$ for any $P \in \Pcal_0$ and any stopping time $\tau$. (In particular, all test supermartingales for $\Pcal_0$ are e-processes by the optional stopping theorem, but not vice versa.)

While $e$-processes can serve as tools to derive \cs{}s, they can also be used as interpretable testing tools in their own right, or as a way to derive anytime $p$-values --- $p$-values that are uniformly valid over time in the same sense as \cs{}s. Formally, an anytime $p$-value for $\Pcal_0$ is an $\Hcal$-adapted process $\infseqt{p_t}$ such that
\begin{equation}\label{eq:anytime-pval}
  \sup_{P \in \Pcal_0} P \left ( \exists t \in \NN : p_t \leq \alpha \right ) \leq \alpha.
\end{equation}
Compare~\eqref{eq:anytime-pval} with a traditional fixed-time $p$-value $p_n$ that satisfies $\forall n \in \NN,\ \sup_{P \in \Pcal_0} P(p_n \leq \alpha) \leq \alpha$. On the other hand, an $e$-process $\infseqt{E_t}$ for $\Pcal_0$ also satisfies Ville's inequality:
\begin{equation}\label{eq:e-proc}
  \sup_{P \in \Pcal_0} P \left ( \exists t \in \NN : E_t \geq 1/\alpha \right ) \leq \alpha.
\end{equation}
As a direct consequence of $\eqref{eq:e-proc}$, notice that $e$-processes yield anytime $p$-values via the transformation $p_t := (1/E_t) \land 1$.
\begin{remark}[Which should you choose: $e$ or $p$?]\label{remark:p-vs-e}
There are several philosophical and practical reasons why one may wish to use $e$-processes over anytime $p$-values, despite the fact that they can both be used for sequential hypothesis testing. Philosophically, $e$-processes (and hence test supermartingales) have game-theoretic interpretations and connections to Bayesian statistics~\citep{shafer2011test,grunwald2019safe,waudby2020confidence}, and have been argued to serve as a better foundation for statistical communication~\citep{shafer2021testing}. Practically, stopped $e$-processes form $e$-values --- nonnegative random variables with expectation at most one \citep{vovk2021values} --- which have several attractive properties over $p$-values, including the fact that they are very straightforward to combine for the sake of testing a global null \citep{vovk2021values}, to perform meta-analyses~\citep{ter2021all}, or to control the false discovery rate under arbitrary dependence~\citep{wang2020false}. In this section, we remain agnostic as to which of the two one should use: we will simply derive $e$-processes and note that their philosophical and practical properties can be enjoyed within OPE, and that if anytime $p$-values are preferred, they are always available via the transformation $p_t := (1/E_t) \land 1$.
\end{remark}
Following \cref{section:policy-value-differences}, let us now derive sequential tests for whether a policy $\pi_1$ has a higher average value than some other policy $\pi_2$. Technically, one could also replace $\pi_1$ or $\pi_2$ with a \emph{sequence} of predictable policies, but for simplicity we will only discuss fixed policies. Concretely, let $\Delta_t$ denote the difference in the values of policies $\pi_1$ and $\pi_2$ at time $t$,
\begin{equation}
  \poldiff_t := \polval_t(\pi_1) - \polval_t(\pi_2) \equiv \EE_{A_t \sim \pi_1}(R_t) - \EE_{A_t \sim \pi_2}(R_t),
\end{equation}
and let $\poldifft := \frac{1}{t} \sum_{i=1}^t \poldiff_t$ denote the running average difference. We are interested in testing the \emph{weak} null hypothesis $H_0$,
\begin{equation}
  H_0 : \forall t,\  \poldifft \leq 0, ~~~\text{vs}~~~H_1 : \exists t :  \poldifft > 0.
\end{equation}
In words, $H_0$ says that ``$\pi_1$ is no better than $\pi_2$ \emph{on average thus far}'' and was used to compare sequential forecasters in~\citet{choe2021comparing}. This ``weak null'' should be contrasted with the ``strong null'' that would posit $H_0^\star : \forall t,\  \poldiff_t \leq 0$ --- clearly, the latter implies the former, and hence any $e$-process (or anytime $p$-value) for $H_0$ can also be used for $H_0^\star$. Mathematically, $H_0$ is a composite superset of the point null $H_0^\star$. From a practical perspective, $H_0$ may be a favorable null to test since it allows for $\polval_t(\pi_1) > \polval_t(\pi_2)$ at various $t$ as long as $\frac{1}{t} \sum_{i=1}^t\polval_t(\pi_1) \leq \frac{1}{t}\sum_{i=1}^t \polval_t(\pi_2)$ for all $t$, whereas $H_0^\star$ requires $\pi_1$ to be uniformly dominated by $\pi_2$.

Let us now derive an explicit $e$-process for $H_0$.
Using the techniques of~\cref{section:policy-value-differences}, define 
\begin{equation}
  \theta_t := w_t^{(1)}R_t - (1-w_t^{(2)}(1-R_t)),
\end{equation}
where $w_t^{(1)} := \pi_1(A_t \mid X_t) / h_t(A_t \mid X_t)$ and $w_t^{(2)} := \pi_2(A_t \mid X_t)/h_t(A_t \mid X_t)$ are the importance weights for policies $\pi_1$ and $\pi_2$. As before, we note that $\theta_t \geq -1$ and $\EE (\theta_t \mid \history_{t-1}) = \poldiff_t$, and hence
\begin{equation}
  \EE \left ( \frac{1}{t}\sum_{i=1}^t \theta_i \Bigm\vert \history_{t-1} \right ) = \poldifft.
\end{equation}
Given the above setup, we are ready to derive an $e$-process (and hence an anytime $p$-value) for the weak null $H_0$ (an illustration is provided in \cref{fig:weak_null_testing}).
\begin{proposition}\label{proposition:testing}
Given contextual bandit data $\infseqt{X_t, A_t, R_t}$ and two target policies $\pi_1$ and $\pi_2$ that we would like to compare, define $S_t(\candpoldifft)$ and $V_t$ by
\begin{align}
  S_t(\candpoldifft) &:= \frac{1}{2}\left ( \sum_{i=1}^t \theta_i - t\candpoldifft \right )~~~\text{and}\\
  V_t &:= \frac{1}{2}\sum_{i=1}^t (\theta_i - \widehat \theta_{i-1})^2, ~~~\text{where}~~\widehat \theta_t := \frac{1}{2} \left [ \left ( \frac{1}{t} \sum_{i=1}^t \theta_i  \right ) \land 1 \right ] 
\end{align}
Then, given any $\rho > 0$, we have that
  \begin{equation}
M_t^\EB(0) := \left(\frac{\rho^\rho e^{-\rho}}{\Gamma(\rho) - \Gamma(\rho, \rho)}\right) \left(\frac{1}{V_t + \rho}\right) \onefone(1, V_t + \rho + 1, S_t(0) + V_t + \rho)
  \end{equation}
  forms an $e$-process for $H_0$. Consequently, $p_t := (1/M_t^\EB) \land 1$ forms an anytime $p$-value for the weak null $H_0: \forall t,\ \poldifft \leq 0$, meaning $\sup_{P \in H_0} P(\exists t : p_t \leq \alpha)\leq \alpha$.
\end{proposition}

\begin{figure}[!htbp]
  \centering
  \includegraphics[width=\textwidth]{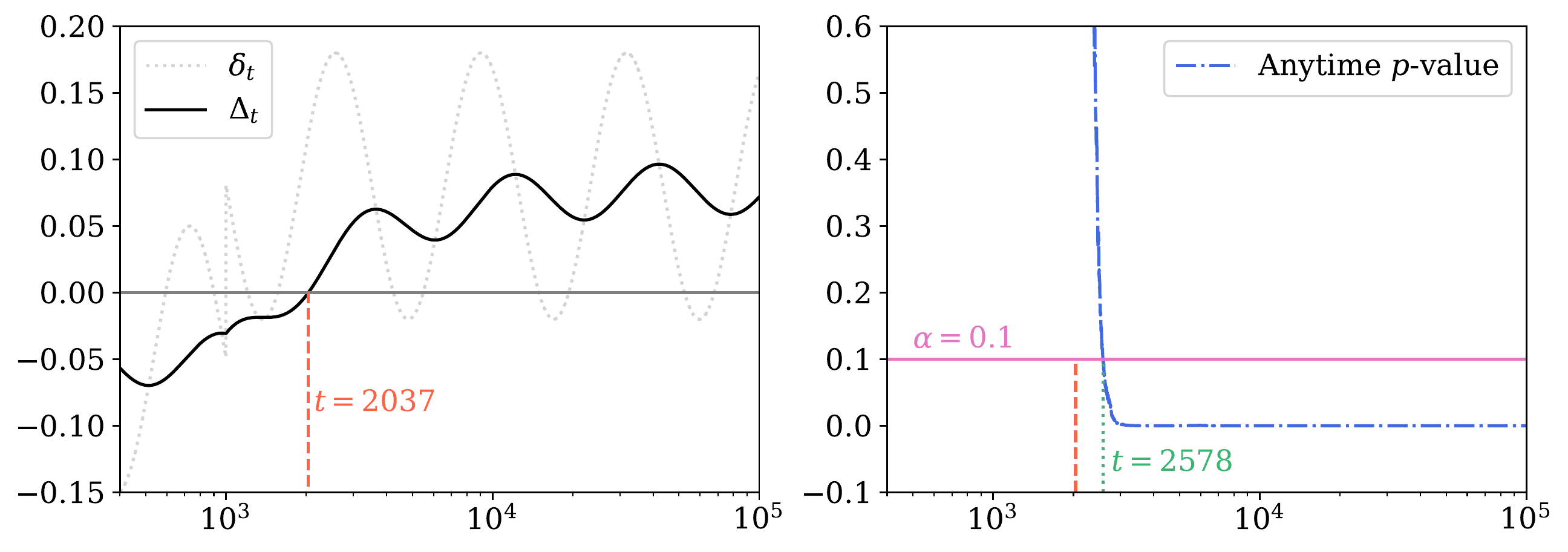}
  \caption{An illustration of how the anytime $p$-value derived in \cref{proposition:testing} can be used to test the weak null $H_0: \forall t,\ \poldifft \leq 0$. In the left-hand side plot, notice that $\poldiff_t$ ventures above 0 at several points prior to $t = 2037$, but the average policy value difference is positive for the first time at $t=2037$. In the right-hand side plot, we see that the anytime $p$-value dips below $\alpha$ shortly after $\poldifft > 0$, at which point the weak null can be safely rejected, with no penalties for the $p$-value having been continuously monitored.}
  \label{fig:weak_null_testing}
\end{figure}
The fact that $M_t^\EB(0)$ forms an $e$-process is easy to see: under $H_0$, we have that $\poldifft \leq 0$ and notice that $M_t^\EB$ with $S_t(0)$ replaced with $S_t(\poldifft)$ forms a test supermartingale using the same techniques as~\cref{section:time-varying}. Since $S_t(\cdot)$ is nonincreasing and $\onefone$ is nondecreasing in its third argument, we have that $M_t^\EB$ is upper-bounded by the aforementioned test supermartingale whenever $\poldifft \leq 0$. The claimed $e$-process property is then an immediate consequence of the optional stopping theorem applied to the above test supermartingale.

\section{Time-uniform inference for the off-policy CDF}\label{section:cdf}
Thus far we have focused on off-policy inference for \emph{mean policy values}, i.e.~functionals of the form $\polval := \EE_{\pi}\left ( R \right )$. In some cases, however, it may be of interest to study quantiles (e.g.~median or 75\textsuperscript{th} percentile) or perhaps the entire cumulative distribution function (CDF) of the reward distribution under policy $\pi$. In this section, we focus on the latter, deriving confidence bands for the CDF $\PP_{\pi}(R \leq r)$ of the reward $R$ under policy $\pi$. Our confidence bands will be uniform in two senses: in time, and in the quantiles. Concretely, if $Q(p)$ and $Q^-(p)$ are the right (standard) and left quantiles, respectively --- meaning $Q(p) := \sup \left \{ x \in \RR : \PP_{\pi}(R \leq x) \right \} $ and $Q^-(p) := \sup \left \{ x \in \RR : \PP_\pi(R < x) \right \}$ --- then we will derive a sequence of confidence bands $[L_t(p), U_t(p)]_{t \in \NN}$ such that
\begin{equation}\label{eq:time-quantile-uniform-example}
  \PP \left ( \forall t \in \NN,p \in (0, 1),\ L_t(p) < Q^-(p)~~\text{and}~~Q(p) < U_t(p) \right ) \geq 1-\alpha.
\end{equation}
Such a guarantee enables anytime-valid inference at arbitrary stopping times for \emph{all quantiles} simultaneously, (as well as any functional thereof). In addition, all our confidence bands will satisfy all five desiderata laid out in \cref{section:desiderata}, and they will consistently shrink to the true quantile $Q(p)$ for all $p$.

In order to state our main result, we first need to define a few terms. Define $W_t, \widebar W_t, \widebar q_t(p)$, and $\ell_t(p; \alpha)$ given by
\begin{alignat}{2}
  W_t &:= \quad &&\sum_{i=1}^t w_i^2 , ~~ \widebar W_t := W_t \lor 1,~~ \\ 
  \widebar q_t(p) &:= \quad &&\logit^{-1} \left ( \logit(p) + 4 \sqrt{\frac{e}{\widebar W_t}} \right ),\\
  \ell_t(p; \alpha) &:= \quad &&2\log \left ( \log \widebar W_t + 1\right ) + 2\log \left ( \left |  \left \lceil \frac{ \sqrt{\widebar W_t}\logit(p)}{4} \right \rceil \right |  \lor 1\right )+\log \left ( \frac{7.06}{\alpha} \right ),\\
  \text{and}~~\boundary_t(p; \alpha) &:= \quad &&\frac{\sqrt{2.13 \ell_t(p; \alpha) \widebar W_t + 1.76 \widebar q_t(p)^2 \ell_t(p; \alpha)^2}}{t} \\
      & \quad && + \frac{1.33 \widebar q_t(p) \ell_t(p; \alpha) + t(\widebar q_t(p) - p)}{t}.
\end{alignat}
While some of the above may seem complicated, they arise naturally from the proof technique discussed below, and it is straightforward to implement them into code. 
Given the above setup, we are ready to state the main result of this section.
\begin{theorem}[Time-uniform confidence band for the off-policy CDF]\label{theorem:cdf}
    Consider a sequence of contextual bandit data $\infseqt{X_t, A_t, R_t}$ with real-valued (i.e.~not necessarily $[0, 1]$-bounded) rewards.
  Let $\widehat F_t^\pi(x) := \frac{1}{t}\sum_{i=1}^t w_i \1(R_i \leq x)$ be the importance-weighted empirical CDF\@, and let $\widehat Q_t(p)$ and $\widehat Q^-_t(p)$ be the upper and lower empirical quantiles, meaning
  \begin{equation}
  \widehat Q_t(p) := \sup \left \{ x \in \RR : \widehat F_t^\pi(x) \leq p \right \}, 
  \end{equation}
  and similarly for $\widehat Q_t^-(p)$ with $\leq$ in the above supremum replaced by a strict inequality $<$.
  Then, 
  \begin{align}
    \PP \left ( \forall t \in \NN, p \in (0, 1),\ Q(p) < \widehat Q_t^- \left ([p + \boundary_t(p; \alpha)]\land 1  \right )  \right ) \geq 1-\alpha.\label{eq:lower-bound-Q(p)}
  \end{align}
  Similarly, after applying the mirroring trick of~\cref{remark:mirroring-trick}, we have that
  \begin{equation}
    \PP \left ( \forall t \in \NN, p \in (0, 1),\ \widehat Q_t \left (\left [p + \frac{1}{t}\sum_{i=1}^t w_i - 1 - \boundary_t(1-p; \alpha) \right ] \lor 0\right ) < Q^-(p)  \right ) \geq 1-\alpha.\label{eq:upper-bound-Q(p)}
  \end{equation}
\end{theorem}

\begin{figure}[!htbp]
  \centering
  \includegraphics[width=0.9\textwidth]{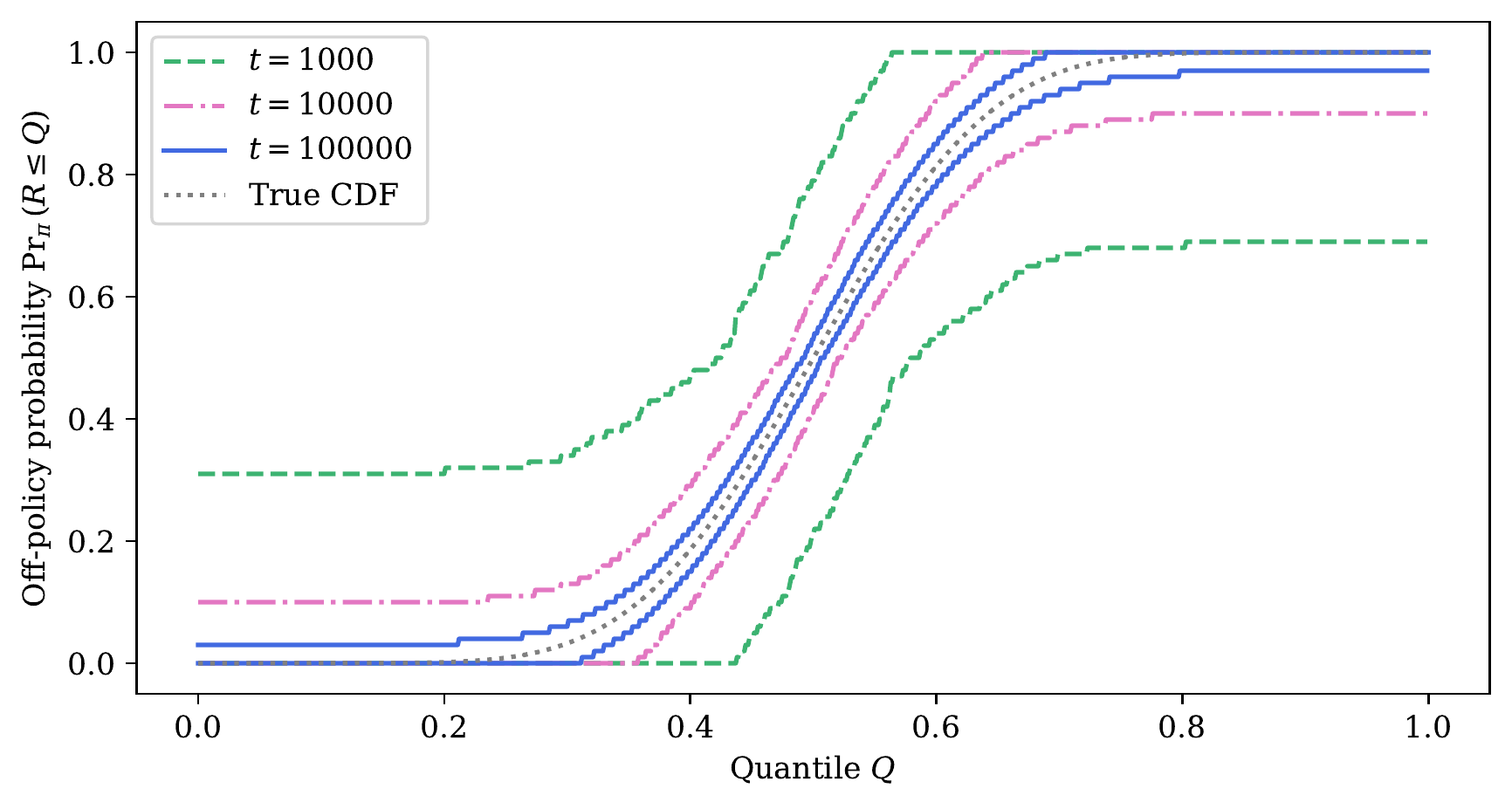}
  \caption{Time- and quantile-uniform 90\% confidence band for the off-policy CDF $\PP_{\pi}(R \leq Q)$ in a Bernoulli(1/2) experiment with the target policy set to $\pi(a \mid x) = \1(a = 1)$ --- i.e. ``always play action 1''. Here, the off-policy distribution of $R$ is a beta distribution with $\alpha=\beta=10$. These confidence bands are simultaneously valid for all $Q \in \RR$ and all $t \in \NN$ (though we only display them at $t \in \{10^3, 10^4, 10^5\}$ above). In particular, notice that as $t$ gets larger, the confidence bands shrink towards the true CDF (and will continue to do so in the limit).}
  \label{fig:cdfs}
\end{figure}

The proof in \cref{proof:cdf} modifies the ``double stitching'' technique of \citet[Theorem 5]{howard2022sequential} to handle importance-weighted observations, and relies on a sub-exponential concentration inequality rather than a sub-Bernoulli one.
Notice that \eqref{eq:lower-bound-Q(p)} and \eqref{eq:upper-bound-Q(p)} could be written without $\widehat Q_t^-$ and $Q^-$ --- i.e.~replacing $\widehat Q_t^-$ and $Q^-$ with $\widehat Q_t$ and $Q$, respectively --- but this would never result in a tighter bound. Illustrations of the time-uniform confidence bands derived in \cref{theorem:cdf} are can be found in \cref{fig:cdfs}.

Many of the \cs{}s throughout this paper have recovered prior \cs{}s in the literature when specialized to the on-policy regime (that is, when all importance weights are set to 1 and reward predictors are set to 0).  Examples include~\cref{theorem:dr-fixed-policy-value} recovering \citet[Theorem 3]{waudby2020estimating} or~\cref{theorem:conjmix-eb} recovering \citet[Proposition 9]{howard2018uniform}. However,~\cref{theorem:cdf} does not recover the on-policy bound it most resembles (\citet[Theorem 5]{howard2022sequential}). The reason for this is subtle, and has to do with the fact that in the on-policy regime, $\1(R_t \leq Q(p)) - p$ is a Bernoulli$(p)$ random variable, hence their partial sums $\frac{1}{t}\sum_{i=1}^t [\1(R_i \leq Q(p)) - p]$ form sub-Bernoulli processes \citep{howard_exponential_2018,howard2018uniform,howard2022sequential} with variance process $tp(1-p)$ \emph{regardless} of the value of $Q(p)$. On the other hand, in the off-policy setting, we use importance-weighted indicators $w_t \1(R_t \leq Q(p)) - p$ whose partial sums are not sub-Bernoulli, but are instead sub-exponential with a variance process that depends on $Q(p)$. This fundamental difference changes the test supermartingales that we have access to, and consequently alters the downstream \cs{}s. 

\paragraph{Comparison with prior work.} There are three prior works that are related to the results of this section, namely those of \citet{howard2022sequential}, \citet{chandak2021universal}, and \citet{huang2021off}, but it is important to note that none of them solve the problem that we are studying --- time-uniform confidence bands for off-policy CDFs --- and hence we focus on a theoretical comparison with these prior works, rather than an empirical one. We discuss each of them in detail below and summarize how they compare with the present paper in~\cref{table:cdf}.

\begin{itemize}
\item \textbf{HR22:} \citet{howard2022sequential} derive time- and quantile-uniform confidence bands for the CDF of iid random variables in the \emph{on-policy} setting, and in particular, Theorem 5 in their paper satisfies a guarantee of the form~\eqref{eq:time-quantile-uniform-example}. However,~\citet{howard2022sequential} do not consider the off-policy inference problem that we do here, and hence the ``Predictable $\infseqt{h_t}$'' and ``$\wmax$-free'' rows are not applicable (N/A). In addition, our setup can be seen as a generalization of theirs if all importance weights are set to 1.

\item \textbf{UnO21:} In the paper entitled ``Universal off-policy evaluation'' (UnO), \citet{chandak2021universal} derive fixed-time quantile-uniform confidence bands for the off-policy CDF\@. Cleverly exploiting monotonicity of CDFs, they reduce the problem to computing finitely many confidence intervals for means of importance-weighted bounded random variables, and taking a union bound over them. Notably, their bounds do not require knowledge of $\wmax$.

  The main difference between our bounds and those of \citet{chandak2021universal} is that ours are both time- and quantile-uniform, while theirs are only quantile-uniform. Note that \citet{chandak2021universal} do consider the more general setup of reinforcement learning in Markov Decision Processes (MDP) --- an area to which we intend to extend all of our \cs{}s in the future --- and MDPs include the contextual bandit setting as a special case. When focusing on contextual bandits specifically, however, and even when ignoring time-uniformity, \cref{theorem:cdf} improves on the fixed-time results of~\citet{chandak2021universal} in the two following ways.

First, the confidence bands of \citet{chandak2021universal} are not guaranteed to shrink to the true CDF as the sample size grows to infinity while ours are (and at an explicit rate of $O(\sqrt{\log t / t})$), which we refer to as ``consistency'' in~\cref{table:cdf}.

Second, their bounds assume that the logging policy $h$ is fixed, while ours can take the form of a sequence of data-dependent logging policies $\infseqt {h_t}$, such as those that result from online learning algorithms. However, since \citet[Theorem 2]{chandak2021universal} simply requires taking a union bound over several \ci{}s for importance-weighted bounded random variables, their Theorem 2 can presumably be extended to handle predictable logging policy sequences by employing the \ci{}s provided in Corollaries~\ref{corollary:betting-ci} or~\ref{corollary:prpl-ci}.

\item \textbf{HLLA21:} \citet{huang2021off} derive impressive quantile-uniform confidence bands for the off-policy CDF\@. Their bounds are elegant and simple to state, resembling the famous Dvoretzky-Kiefer-Wolfowitz (DKW) inequalities and their sharpened forms \citep{dvoretzky1956asymptotic,massart1990tight}. Notably, their bounds are consistent for the true CDF at $O(1/\sqrt{n})$ rates. Similar to \citet{chandak2021universal}, however, their results are not time-uniform, and hence do not permit valid inference at stopping times, unlike ours. Moreover, all of their bounds require knowledge of $\wmax$ (and for this value to be finite), while our bounds do not. Finally, similar to \citet{chandak2021universal}, their bounds assume that the logging policy $h$ is fixed.

\end{itemize}

\begin{table}[!htbp]
  \caption{Comparison of various uniform confidence bands for the CDF}
  \label{table:cdf}
  \centering
  \begin{tabular}{|l|c|c|c|c|}
    \hline
    & HR22 & UnO21 & HLLA21 & Thm.~\ref{theorem:cdf} \\
    \hline
Off-policy &  &\checkmark &\checkmark & \checkmark \\
Time-uniform & \checkmark & & & \checkmark \\
Consistency & \checkmark &  & \checkmark & \checkmark \\
Predictable $\infseqt{h_t}$ & N/A & & & \checkmark \\
$\wmax$-free & N/A & \checkmark & & \checkmark \\
    \hline
  \end{tabular}
\end{table}







\section{Summary \& extensions}\label{section:summary}

This paper derived time-uniform confidence sequences for various parameters in off-policy evaluation which remain valid even in contextual bandit setups where data are collected adaptively and sequentially over time. We began in~\cref{section:desiderata} by laying out our desiderata for off-policy inference: we sought methods that (1) are exact and nonasymptotically valid, (2) only make nonparametric assumptions such as boundedness, (3) are time-uniform, and hence valid at arbitrary stopping times, (4) do not require knowledge of extreme values of importance weights, and (5) allow data to be collected by data-dependent logging policies.

In~\cref{section:warmup}, we began by studying the most classical off-policy parameter --- a fixed policy value $\polval$ --- and we derived \cs{}s that strictly generalize prior state-of-the-art \cs{}s by weakening the required assumptions and allowing for variance reduction via double robustness. In the same section, we also develop the first closed-form confidence sequences for policy values, as well as some tight fixed-time confidence intervals that are instantiations of our time-uniform bounds. \cref{section:time-varying} then developed \cs{}s for a more general parameter: the time-varying average policy value $\polvalt := \frac{1}{t}\sum_{i=1}^t \polval_i$, and we discussed what implications these bounds have for adaptive sequential experiments, such as online A/B tests.

Finally, in \cref{section:cdf}, we derived simultaneously valid \cs{}s for every quantile of the off-policy reward distribution. Said differently, these bounds form time-uniform confidence bands for the CDF of the off-policy reward distribution.

There are a few other works that consider \cs{}s and test supermartingales without reference to OPE or contextual bandits, but that nevertheless now have interesting consequences in the OPE problem once paired with the present paper. In particular, we want to highlight the implications that~\citet{wang2020false} and~\citet{xu2022post} have on false discovery rate control in OPE\@, and how~\citet{waudby2022locally} immediately yields algorithms for differentially private OPE\@. We briefly discuss these implications below, but omit their full derivations since these extensions are rather simple and not central to the current paper.

\paragraph{False discovery/coverage rate control under arbitrary dependence.}

Suppose that rather than estimate a single policy value $\polval$, we are interested in a \emph{collection} $(\polval_1, \dots, \polval_J)$ containing the values of the policies $(\pi_1, \dots, \pi_J)$. When testing several hypothesis or constructing several \ci{}s, etc., it is often of interest to control some multiple testing metric, such as the false discovery rate (FDR), or the false coverage rate (FCR), respectively~\citep{benjamini1995controlling,benjamini2005false}. Rather surprisingly,~\citet{wang2020false} and~\citet{xu2022post} show that for tests and \ci{}s built from \emph{$e$-values} --- nonnegative test statistics with expectation at most one --- the FDR and FCR can be controlled under \emph{arbitrary} dependence with virtually no modification to the famous Benjamini-Hochberg~\citep{benjamini1995controlling} and Benjamini-Yekutieli~\citep{benjamini2005false} procedures, while this fact is not true for generic tests and \ci{}s. Relevant to the current paper, \emph{all} of our \cs{}s are fundamentally built from test supermartingales which form $e$-values at arbitrary stopping times. As a concrete consequence, we can take a collection of stopped \cs{}s $C_\tau^{(1)}, \dots, C_\tau^{(J)}$ for $(\pi_j)_{j \in [J]}$, adjust them via the $e$-BY procedure of~\citet{xu2022post} to produce $\widetilde C_\tau^{(1)}, \dots, \widetilde C_\tau^{(J)}$, so that the FCR is controlled at some desired level $\delta \in (0, 1)$. The ability control the FCR under arbitrary dependence is crucial for our setting since the \cs{}s $(C_\tau^{(j)})_{j\in[J]}$ are highly dependent and constructed from the same data, but with different importance weights. Similar implications hold for sequential tests and control of the FDR via the $e$-BH procedure of~\citet{wang2020false}.

\paragraph{Locally differentially private off-policy evaluation in contextual bandits.}
\citet{waudby2022locally} developed nonparametric \cs{}s and \ci{}s for means of bounded random variables under privacy constraints. The authors developed a so-called ``Nonparametric randomized response'' (NPRR) mechanism that serves as a nonparametric generalization of Warner's randomized response \citep{warner1965randomized}, mapping a $[0, 1]$-bounded random variable $Y_t$ to a new random variable $Z_t$ so that each $Z_t$ is an $\varepsilon$-locally differentially private view of $Y_t$ with mean $r\EE(Y_t) + (1-r)/2$, where $r$ is a known quantity that depends on $\varepsilon$ (and hence it is possible to work out what $\EE(Y_t)$ is). While~\citet{waudby2022locally} did not explicitly consider the contextual bandit setup, they did develop \cs{}s for time-varying treatment effects in sequential experiments, similar to the discussion in \cref{section:ate-randomized-expts}. However, like other prior work, their \cs{}s require \emph{a priori} knowledge of the minimal propensity score (in the language of this paper: they require knowledge of $\wmax$). Nevertheless, it is possible to derive locally private \cs{}s for (time-varying) policy values without knowledge of $\wmax$ using the techniques of the current paper. Moreover, several policies can be evaluated from a \emph{single} application of NPRR, thereby avoiding inflation of the privacy parameter $\varepsilon$ from evaluating multiple policies. That is, given $[0, 1]$-bounded rewards $\infseqt{R_t}$, we can use NPRR to generate private views $\infseqt{Z_t}$ of these rewards, and notice that $\EE(w_tZ_t) = \EE_{A_t \sim \pi}(Z_t) = r\EE_{A_t\sim\pi}(R_t) + (1-r)/2$, and hence a \cs{} for $\EE(w_tZ_t)$ can be translated into a \cs{} for $\EE_{A_t \sim \pi}(R_t)$ even though we only see a privatized version of $R_t$. In particular, practitioners can derive locally private \cs{}s for time-varying policy values using \cref{theorem:conjmix-eb} for several policies $\pi_1, \dots, \pi_J$, with only a single application of NPRR.

\bigskip

We believe that this paper presents a comprehensive treatment of OPE inference, yielding procedures that are theoretically valid under more general settings and yet deliver state-of-the-art practical performance. A challenging open problem is to extend these techniques to the off-policy MDP (Markov Decision Process) setting, where the actions at each step affect subsequent covariate and reward distributions, as captured by state variables. Another important open problem is to design practical OPE inference methods not just for one policy, but uniformly over an entire family of policies. 

\subsection*{Acknowledgements}
IW-S thanks Martin Larsson, Alec McClean, Steve Howard, Ruohan Zhan, and Edward H. Kennedy for helpful discussions. The authors acknowledge support from NSF grants IIS-2229881 and DMS-2310718.
AR acknowledges funding from NSF Grant DMS2053804 and from ARL IoBT REIGN. Research reported in this paper was sponsored in part by the DEVCOM Army Research Laboratory under Cooperative Agreement W911NF-17-2-0196 (ARL IoBTCRA). The views and conclusions contained in this document are those of the authors and should not be interpreted as representing the official policies, either expressed or implied, of the Army Research Laboratory or the U.S. Government. The U.S. Government is authorized to reproduce and distribute reprints for Government purposes notwithstanding any copyright notation herein.

\bibliographystyle{plainnat}
\bibliography{references.bib}
\newpage
\appendix

\section{Proofs of the main results}

\subsection{A technical lemma}

\begin{lemma}\label{lemma:bernsm}
Let $Z$ and $\widehat Z$ be $\history$-adapted processes such that $Z_t - \widehat Z_{t-1} \geq -1$ almost surely for all $t$. Denoting $\mu_t := \EE(Z_t \mid \history_{t-1})$, we have that for any $(0, 1)$-valued predictable process $\infseqt{\lambda_t}$, 
\begin{equation}
M_t := \exp\left(\sum_{i=1}^t \lambda_i (Z_t - \mu_t) - \sum_{i=1}^t \psi_E(\lambda_i) \left(Z_t - \widehat Z_{t-1}\right)^2\right),
    \label{eq:ebernsm}
\end{equation}
forms a test supermartingale, where $\psi_E(\lambda) := -\log(1-\lambda) - \lambda$.
\end{lemma}
Above, $\widehat Z_{t-1}$ is to be interpreted as an estimator of $Z_t$ using the first $t-1$ samples. Closely related lemmas have appeared in~\cite{fan2015exponential,howard_exponential_2018,howard2018uniform,waudby2020estimating}, but those papers assumed $Z_t-\mu_t \geq -1$, which does not suffice for our purposes. What is somewhat surprising above is that we do not require a particular lower bound on $Z_t$ or an upper bound on $\mu_t$, as long as $Z_t - \widehat Z_{t-1} \geq -1$.

\begin{proof}
First, note that $M_0 \equiv 1$ by construction, and $M_t$ is always positive. It remains to show that $M_t$ forms a supermartingale. Writing out the conditional expectation of $M_t$ given $\history_{t-1}$, we have that
\begin{align}
    \EE(M_t \mid \history_{t-1}) = M_{t-1} \underbrace{\EE \left (\exp \left \{ \lambda_t (Z_t - \mu_t) - \psi_E(\lambda_t) (Z_t - \widehat Z_{t-1})^2\right \} \mid \history_{t-1} \right )}_{(\dagger)},
\end{align}
and hence it suffices to prove that $(\dagger) \leq 1$.
Denote for the sake of succinctness,
\[Y_{t} := Z_t - \mu_t ~~~\text{ and }~~~ \delta_t := \widehat Z_{t-1} - \mu_t, \]
and note that $\EE(Y_t \mid \history_{t-1}) = 0$.
Using the proof of \citet[Proposition 4.1]{fan2015exponential}, we have that $\exp\{b\lambda  - b^2 \psi_{E} (\lambda)\} \leq 1 + b \lambda$ for any $\lambda \in [0, 1)$ and $b \geq -1$. Setting $b := Y_{t} - \delta_{t} = Z_t- \widehat Z_{t-1}$,
\begin{align*}
	&\EE \left [\exp \left \{ \lambda_t Y_{t} - (Y_{t} - \delta_{t})^2 \psi_{E}(\lambda_t ) \right \}  \Bigm | \history_{t-1} \right ]\\
	=\ &\EE \left [\exp \left \{ \lambda_t (Y_{t}-\delta_{t}) - (Y_{t} - \delta_{t})^2 \psi_{E}(\lambda_t) \right \}  \bigm | \history_{t-1} \right ] \exp(\lambda_t \delta_{t})\\
  \leq \ &\EE\left  [1 + (Y - \delta_{t} )\lambda_t \mid \history_{t-1} \right ]\exp(\lambda_t \delta_{t}) \\
  = \ &\EE\left [1 - \delta_{t} \lambda_t \mid \history_{t-1} \right ]\exp(\lambda_t \delta_{t}) \leq 1,
\end{align*}
where the last line follows from the fact that $Y_{t}$ is conditionally mean zero
and the inequality $1-x \leq \exp(-x)$ for all $x \in \mathbb R$. This completes the proof.
  
\end{proof}

\subsection{Proof of \cref{theorem:dr-fixed-policy-value}}\label{proof:dr-fixed-policy-value}

We will only derive the lower \cs{} for $\polval$, since the upper \cs{} follows analogously. Consider the process $\infseqt{M_t(\nu)}$ given by
\begin{equation}\label{eq:label}
  M_t(\nu) :=  \prod_{i=1}^t \left [ 1 + \lambda_i^L(\nu) \cdot (\phi_i^\DRL - \nu) \right ]. 
\end{equation}

The proof proceeds in three steps, following the strategy of~\cite{howard2018uniform,waudby2020estimating} and~\cite{karampatziakis2021off}. 
In Step 1, we show that the pseudo-outcomes have conditional mean $\polval$, i.e. $\EE(\phi_t^\DRL \mid \history_{t-1}) = \EE_\pi (R_t \mid \history_{t-1}) =\polval$. In Step 2, we use Step 1 to show that $M_t(\polval)$ forms a test martingale and apply Ville's inequality to it. In Step 3, we ``invert'' this test martingale to obtain the lower \cs{} found in \cref{theorem:dr-fixed-policy-value}.

\paragraph{Step 1: Computing the conditional mean of the doubly robust pseudo-outcomes.}
Writing out the conditional expectation of $\phi_t^\DRL$, we have
\begin{align*}
  &\EE[\phi_t^\DRL \mid \history_{t-1}] \\
  = \ &\EE (w_t R_t \mid \history_{t-1}) - \EE \left \{ w_t \cdot \left (\widehat r_t(X_t; A_t) \land \frac{k_t}{w_t} \right ) - \EE_{a \sim \pi(\cdot \mid X_t)} \left ( \widehat r_t(X_t; a) \land \frac{k_t}{w_t} \right ) \mid \history_{t-1} \right \} \\ = \ & \EE (w_t R_t \mid \history_{t-1}) \\
  =\ &\int_{(x,a,r)} \frac{\pi(a \mid x)}{h_t(a\mid x)} r \cdot p_{R_t} (r \mid a, x, \history_{t-1}) h_t(a \mid x) p_{X_t}(x \mid \history_{t-1}) \dd x \dd a \dd r\\
  =\ &\int_{(x,a,r)} r \cdot p_{R_t} (r \mid a, x, \history_{t-1}) \pi(a \mid x) p_{X_t}(x \mid \history_{t-1}) \dd x \dd a \dd r\\
  = \ &\EE_\pi (R_t \mid \history_{t-1}) = \polval.
\end{align*}

\paragraph{Step 2: Showing that $M_t(\polval)$ forms a test martingale.}
 First, note that $M_0 \equiv 1$ by construction. To show that $M_t$ is nonnegative, notice that since $R_t, \widehat r_t \in [0, 1]$ almost surely, we have that $\phi_t^\DRL \geq -k_t$. Therefore, for any $\nu \in [0, 1]$, 
\begin{align*}
    1 + \lambda_t^L(\nu) \cdot (\phi_t^\DRL - \nu) 
   \geq \ & 1 + \lambda_t^L(\nu) \cdot (- k_t - \nu) \\
   > \ & 1 + \frac{- k_t - \nu}{k_t + \nu} \ \ \ (\text{since $\lambda_t^L(\nu) \in [0, (\nu + k_t)^{-1})$}) \\
   = \ & 0.
\end{align*}
Lastly, it remains to show that $\EE \left [ M_t(\polval) \mid \history_{t-1} \right ] = M_{t-1}(\polval)$. Writing out the conditional expectation of $M_t(\polval)$, we have
\begin{align*}
    \EE \left [ M_t(\polval) \mid \history_{t-1} \right ] & = \EE \left [ M_{t-1} (\polval) \left \{ 1 + \lambda_t^L(\polval) \cdot (\phi_t^\DRL - \polval) \right \} \mid \history_{t-1} \right ] \\
    & =   M_{t-1}(\polval) \cdot \left [ 1 + \lambda_t^L(\polval) \cdot \EE \left \{ (\phi_t^\DRL - \polval)  \mid \history_{t-1} \right \} \right ] \\
    & =  M_{t-1}(\polval) \cdot (1 + \lambda_t^L(\polval) \cdot 0) = M_{t-1}(\polval), 
\end{align*}
where the second line follows from the fact that $M_{t-1}, \lambda_t^L$ are predictable, and the third line follows from Step 1. Therefore, by Ville's inequality for nonnegative supermartingales \citep{ville1939etude}, we have
\begin{equation}\label{eq:ville}
    \PP \left (\exists t \in \NN,\ M_t(\polval) \geq \frac{1}{\alpha} \right ) \leq \alpha.
\end{equation}

\paragraph{Step 3: Inverting Ville's inequality to obtain a lower \cs{}.} Recall the lower boundary given by \eqref{eq:dr-lower},
\begin{equation*}
    L_t^\DR := \inf \left \{ \candpolval \in [0, 1] : \prod_{i=1}^t \left [ 1 + \lambda_i^L(\candpolval) \cdot (\phi_i^\DRL - \candpolval) \right ] < \frac{1}{\alpha} \right \}
\end{equation*}
and notice that if $\polval < L_t^\DR$, then  $M_t(\polval) \geq 1/\alpha $ by definition of $L_t^\DR$. Consequently,
\begin{equation*}
 \PP(\exists t \in \NN,\ \polval < L_t^\mathrm{DR}) \leq \PP \left (\exists t \in \NN,\ M_t(\polval) \geq \frac{1}{\alpha} \right ) \leq \alpha.
\end{equation*}
 Therefore, we have $\PP(\forall t \in \NN,\ \polval \geq L_t^\mathrm{DR}) \geq 1-\alpha$, so $L_t^\DR$ forms a lower $(1-\alpha)$-\cs{} for $\polval$, which completes the proof.

\qed

\subsection{Proof of \cref{theorem:conjmix-eb}}\label{proof:conjmix-eb}



\begin{proof}[Proof of \cref{theorem:conjmix-eb}]
The proof proceeds in three steps, following the high level outline of the conjugate mixture method in~\cite{howard2018uniform}. First, we invoke \cref{lemma:bernsm} to derive a test supermartingale for each $\lambda \in (0, 1)$. Second, we mix over $\lambda \in (0, 1)$ using the truncated gamma density to obtain \eqref{eq:gamma-exponential-nsm}. Third and finally, we invert this test supermartingale to obtain a lower \cs{} for $\polvalt$.
\paragraph{Step 1: Deriving a test supermartingale indexed by $\lambda \in (0, 1)$.}
    Let $Z_t := \xi_t$ and $\widehat Z_{t-1} := \widehat \xi_{t-1}$ as in the setup of \cref{theorem:conjmix-eb}. First, notice that $\EE(\xi_t \mid \history_{t-1}) = \polval_t$:
    \begin{align}
        \EE(\xi_t \mid \history_{t-1}) &= \EE(w_t R_t \mid \history_{t-1}) \\
        &= \int_{x,a,r} \frac{\pi(a \mid x)}{h_t(a \mid x)} r p_{R_t}(r \mid a, x, \history_{t-1}) h_t(a \mid x) p_{X_t}(x \mid \history_{t-1}) \dd x \dd a \dd r .
    \end{align}
    Notice that
    \(
        \xi_t - \widehat \xi_{t-1} \geq -1,
    \)
    and hence by \cref{lemma:bernsm}, we have that for any $\lambda \in (0, 1)$
    \begin{equation*}
        M_t(\polvalt; \lambda) := \exp \left \{ \lambda S_t(\polvalt) - V_t \psi_E(\lambda) \right \}
    \end{equation*}
    forms a test supermartingale.
    
\paragraph{Step 2: Mixing over $\lambda$ using the truncated gamma density.}
For any distribution $F$ on $(0, 1)$,
\begin{align}
    M_t^\EB(\polvalt) := \int_{\lambda \in (0, 1)} M_t(\polvalt; \lambda) \dd F(\lambda)
\end{align}
forms a test supermartingale by Fubini's theorem. In particular, we will use the truncated gamma density $f(\lambda)$ given by
\begin{equation}
f(\lambda) = \frac{\rho^\rho e^{-\rho \left(1 - \lambda\right)} \left(1 - \lambda\right)^{\rho - 1}}{\Gamma(\rho) - \Gamma(\rho, \rho)}, 
\end{equation}
as the mixing density. Writing out $M_t(\polval)$ using $dF(\lambda) := f(\lambda)d\lambda$, we have 
\begin{align*}
M_t^\EB(\polvalt) &:= \int_0^1 \exp \left \{ \lambda S_t(\polvalt) - V_t \psi_E(\lambda) \right \} f(\lambda) \dd\lambda \\
&=\int_0^1 \exp \left \{ \lambda S_t(\polvalt) - V_t \psi_E(\lambda) \right \} \frac{\rho^\rho e^{-\rho \left(1 - \lambda\right)} \left(1 - \lambda\right)^{\rho - 1}}{\Gamma(\rho) - \Gamma(\rho, \rho)} \dd\lambda \\
&= \frac{\rho^\rho e^{-\rho}}{\Gamma(\rho) - \Gamma(\rho, \rho)} \int_0^1 \exp \{{\lambda \left(\rho + S_t + V_t\right)}\} \left(1 - \lambda\right)^{V_t + \rho - 1} \dd\lambda \\
&= \left(\frac{\rho^\rho e^{-\rho}}{\Gamma(\rho) - \Gamma(\rho, \rho)}\right) \left(\frac{1}{V_t + \rho}\right) \left.\left(  \frac{\Gamma(b)}{\Gamma(a) \Gamma(b-a)} \int_0^1 e^{z u} u^{a-1} (1 - u)^{b - a - 1} \dd u \right)\right|_{\substack{a = 1 \\ b = V_t + \rho + 1 \\ z = S_t + V_t + \rho }} \\
&= \left(\frac{\rho^\rho e^{-\rho}}{\Gamma(\rho) - \Gamma(\rho, \rho)}\right) \left(\frac{1}{V_t + \rho}\right) \onefone(1, V_t + \rho + 1, S_t + V_t + \rho),
\end{align*}
which completes this step.

\paragraph{Step 3: Inverting the mixture test supermartingale to obtain \eqref{eq:gamma-exponential-nsm}.}
Similar to Step 3 of the proof of \cref{theorem:dr-fixed-policy-value}, we have that $\polvalt < L_t^\EB$ if and only if $M_t(\polvalt) \geq 1/\alpha$, and hence by Ville's inequality for nonnegative supermartingales, we have that
\begin{equation*}
    \PP(\exists t : \polvalt < L_t^\EB) = \PP(\exists t : M_t^\EB(\polvalt) \geq 1/\alpha) \leq \alpha,
\end{equation*}
and hence $L_t^\EB$ forms a lower $(1-\alpha)$-\cs{} for $\polvalt$. This completes the proof.
\end{proof}

\begin{remark}[Writing \eqref{eq:gamma-exponential-nsm} in terms of the lower incomplete gamma function]\label{section:onefone-vs-lowerincgamma}
For readers familiar with \citet[Proposition 9]{howard2018uniform}, we can rewrite \eqref{eq:gamma-exponential-nsm} in terms of the lower incomplete gamma function via the identity
$\onefone(1, b, z) = (b-1) e^z z^{1-b} (\Gamma(b-1)-\Gamma(b-1,z))$, resulting in \begin{align*}
& \left(\frac{\rho^\rho e^{-\rho}}{\Gamma(\rho) - \Gamma(\rho, \rho)}\right) \left(\frac{1}{v + \rho}\right) \onefone(1, v + \rho + 1, s + v + \rho) \\
 = \ &\left(\frac{\rho^\rho}{\Gamma(\rho) \gamma(\rho, \rho)}\right)  \frac{\Gamma(v + \rho) \gamma(v + \rho, s + v + \rho)}{(s + v + \rho)^{v + \rho}} \exp\left(s + v\right),
\end{align*}
where $\gamma(\cdot, \cdot)$ is the lower regularized incomplete gamma function and $v = V_t$ and $s = S_t(\polvalt)$.  This matches \citet[Eq. (66)]{howard2018uniform} when setting $c = 1$. The final representation above is real-valued after some complex terms are cancelled (in the case where $(s + v + \rho)$ is negative), but the representation in terms of $\onefone(1, \cdot, \cdot)$ sidesteps this subtlety altogether, which is why we prefer to use it in \cref{theorem:conjmix-eb}.

\end{remark}

\subsection{Proof of \cref{proposition:prpl-cs}}\label{proof:prpl-cs}
\begin{proof}
  Consider the process $M \equiv \infseqt{M_t}$ given by
  \begin{equation}
    M_t := \exp \left \{ \sum_{i=1}^t \lambda_i \left ( \xi_i - \frac{\polval}{k_i + 1} \right ) - \sum_{i=1}^t \left ( \xi_i - \widehat \xi_{i-1} \right )^2\psi_E(\lambda_i)  \right \}.
  \end{equation}
  Then by~\cref{lemma:bernsm}, we have that $M$ is a test supermartingale, and hence by Ville's inequality, $\PP(\exists t : M_t \geq 1/\alpha)\leq \alpha$. Inverting this time-uniform concentration inequality, we have that with probability at least $(1-\alpha)$ and for all $t \in \NN$,
  \begin{align*}
    M_t < 1/\alpha &\iff \exp \left \{ \sum_{i=1}^t \lambda_i \left ( \xi_i - \frac{\polval}{k_i + 1} \right ) - \sum_{i=1}^t \left ( \xi_i - \widehat \xi_{i-1} \right )^2\psi_E(\lambda_i)  \right \} < \frac{1}{\alpha}\\
        &\iff \sum_{i=1}^t \lambda_i \left ( \xi_i - \frac{\polval}{k_i + 1} \right ) - \sum_{i=1}^t \left ( \xi_i - \widehat \xi_{i-1} \right )^2 \psi_E(\lambda_i) < \log(1/\alpha)\\
        &\iff \sum_{i=1}^t \lambda_i\xi_i - \polval \sum_{i=1}^t \frac{\lambda_i}{k_i+1} - \sum_{i=1}^t \left ( \xi_i - \widehat \xi_{i-1} \right )^2 \psi_E(\lambda_i) < \log(1/\alpha)\\
    &\iff \polval > \frac{\sum_{i=1}^t \lambda_i \xi_i}{\sum_{i=1}^t \lambda_i / (k_i + 1)} - \frac{\log(1/\alpha) + \sum_{i=1}^t \left ( \xi_i - \widehat \xi_{i-1} \right )^2\psi_E(\lambda_i)}{\sum_{i=1}^t \lambda_i / (k_i + 1)},
  \end{align*}
  which completes the proof.
\end{proof}

\subsection{Proof of \cref{proposition:lil-eb}}\label{proof:lil-eb}

We will prove a more general result below for arbitrary $\eta, s > 1$, but the exact constants in \cref{proposition:lil-eb} can be obtained by setting $\eta = e$, $s = 2$. By~\cref{lemma:bernsm} combined with \citet[Table 5, row 7]{howard_exponential_2018}, we have that $S_t(\polvalt)$ is a sub-gamma process with scale parameter $c = 1$, meaning for any $\lambda \in [0, 1)$
\begin{equation}
  M_t^G(\lambda) := \exp \left \{ \lambda S_t(\polvalt) - V_t \psi_{G}(\lambda) \right \},
\end{equation}
where $\psi_G(\lambda) \equiv \psi_{G,1}(\lambda) = \frac{\lambda^2 }{ 2(1-\lambda) }$. Define the following parameters:
\begin{align*}
  \lambda_k &:= \psi^{-1}(\log(1/\alpha) / \eta^{k+1/2}), ~~\text{where}~~\psi^{-1}_G(a) := \frac{2}{1 + \sqrt{1 + 2/a}},\\
  \alpha_k &:= \frac{\alpha}{(k+1)^s \zeta(s)},~~\text{and}\\
  b_{t,k} &:= \frac{V_t \psi_G (\lambda_k) + \log(1/\alpha_k)}{\lambda_k}.
\end{align*}
Taking a union bound over $k \in \NN$, we have that
\begin{equation*}
  \PP \left ( \forall t \in \NN, k \in \NN,\ S_t(\polvalt) \leq b_{t,k} \right ) \geq 1-\alpha.
\end{equation*}
It remains to find a deterministic upper bound on $b_{t,k}$ that does not depend on $k$. Indeed, similar to \citet[Eq. (39)]{howard2018uniform}, we have that
\begin{equation}
  b_{t, k} = A \left ( \frac{\log(1/\alpha_k)}{\eta^{k+1/2}} \right ) \underbrace{\left [ \sqrt{\frac{\eta^{k+1/2}}{V_t}} + \sqrt{\frac{V_t}{\eta^{k+1/2}}} \right ]}_{(\star)} \sqrt{\frac{\log(1/\alpha_k) V_t}{2}},
\end{equation}
where $A(a) := \sqrt{2a} / \psi^{-1}_G(a) = \sqrt{1 + a/2} + \sqrt{a/2}$. Now, notice that $(\star)$ is convex in $V_t$ on $V_t \in [\eta^k, \eta^{k+1}]$, and hence $(\star)$ is maximized at the endpoints $\eta^{k}$ and $\eta^{k+1}$. Consequently, on the $k^\mathrm{th}$ epoch --- i.e.~when $\eta^k \leq V_t \leq \eta^{k+1}$ --- we have that
\begin{align}
  b_{t,k} &\leq A \left ( \frac{\log(1/\alpha_{k})}{\eta^{k+1/2}}  \right ) \left [ \eta^{-1/4} + \eta^{1/4} \right ] \sqrt{\frac{\log(1/\alpha_k) V_t}{2}} \nonumber \\
          &\leq A \left ( \frac{\sqrt{\eta}\log(1/\alpha_{k})}{V_t}  \right ) \left [ \eta^{-1/4} + \eta^{1/4} \right ] \sqrt{\frac{\log(1/\alpha_k) V_t}{2}} \nonumber \\
          &= \left [ \sqrt{1+\frac{\sqrt{\eta} \log(1/\alpha_k)}{2V_t}} + \sqrt{\frac{\sqrt{\eta}\log(1/\alpha_k)}{2V_t}} \right ] \cdot \left [ \eta^{-1/4} + \eta^{1/4} \right ] \cdot \sqrt{\frac{\log(1/\alpha_k)V_t}{2}} \label{eq:proof-lil-eb-line1}
\end{align}
where the first inequality follows from our analysis of $(\star)$, the second follows from monotonicity of $A(\cdot)$ and the fact that $V_t \leq \eta^{k+1}$ on the $k^\mathrm{th}$ epoch, the third follows from the definition of $A(\cdot)$. Rewriting the final line~\eqref{eq:proof-lil-eb-line1} more succinctly, we have the following upper bound on $b_{t,k}$ for every $k \in \NN$,
\begin{align}
  &b_{t,k} \leq \sqrt{\gamma_1^2 \log(1/\alpha_k) V_t + \gamma_2^2 \log^2(1/\alpha_k)} + \gamma_2 \log(1/\alpha_k), \label{eq:proof-lil-eb-line2}\\
  &\text{where}~~\gamma_1 := \frac{\eta^{1/4} + \eta^{-1/4}}{\sqrt{2}},~~\text{and}~~\gamma_2 := \frac{\sqrt{\eta} + 1}{2}.
\end{align}
Now, notice that the above bound only depends on $k$ through $\log(1/\alpha_k)$. As such, we will upper bound $\log(1/\alpha_k)$ solely in terms of $V_t$ and other constants. Indeed, on the $k^\mathrm{th}$ epoch, we have
\begin{align}
  \log(1/\alpha_k) &\equiv \log \left ( \frac{(k+1)^s \zeta(s)}{\alpha} \right ) \nonumber 
                   = s\log \left ( k+1 \right ) + \log \left ( \frac{\zeta(s)}{\alpha} \right ) \nonumber \\
  & \leq s \log \left ( \log_\eta V_t + 1 \right ) + \log \left ( \frac{\zeta(s)}{\alpha} \right ) \equiv \ell_t\label{eq:proof-lil-eb-upper-bound-alphak}
\end{align}
where the final line used the upper bound $k \leq \log_\eta V_t $ which follows because $\eta^k \leq V_t$ on the $k^\mathrm{th}$ epoch. Combining~\eqref{eq:proof-lil-eb-line2} and~\eqref{eq:proof-lil-eb-upper-bound-alphak}, we have that
\begin{equation}
  b_{t,k} \leq \sqrt{\gamma_1^2 \ell_t V_t + \gamma_2^2 \ell_t^2} + \gamma_2 \ell_t, ~~\text{where}~~\ell_t := s \log \left ( \log_\eta V_t + 1 \right ) + \log \left ( \frac{\zeta(s)}{\alpha} \right ),
\end{equation}
which no longer depends on $k$. Consequently, we have that
\begin{align*}
  1-\alpha &\leq \PP \left ( \forall t \in \NN,\ \sum_{i=1}^t \xi_i - \sum_{i=1}^t \polval_i \leq \sqrt{\gamma_1^2 \ell_t V_t + \gamma_2^2 \ell_t^2} + \gamma_2 \ell_t \right ) \\
  & = \PP \left ( \forall t \in \NN,\ \polvalt \geq \underbrace{\frac{1}{t} \sum_{i=1}^t \xi_i - \frac{\sqrt{\gamma_1^2 \ell_t V_t + \gamma_2^2 \ell_t^2}}{t} - \frac{\gamma_2 \ell_t}{t}}_{(\dagger)} \right ),
\end{align*}
and hence $(\dagger)$ forms a lower $(1-\alpha)$-\cs{} for $\polvalt$.

\subsection{Proof of \cref{theorem:cdf}}\label{proof:cdf}



We will prove a more general result below for arbitrary $\eta, s, \delta > 1$, but the exact constants in \cref{proposition:lil-eb} can be obtained by setting $\eta = e$, and $s = \delta = 2$.
The proof will proceed in five steps. First, we derive an exponential $e$-process --- i.e.~an adapted process upper-bounded by a test supermartingale --- from $S_t(p) := \sum_{i=1}^t w_i \1(R_i \leq Q(p)) - tp$. Second, we apply Ville's inequality to the aforementioned $e$-process to obtain a level-$\alpha$ linear boundary $b_t(p)$ on $S_t(p)$, meaning $\PP(\exists t \in \NN : S_t(p) \geq b_t(p)) \leq \alpha$. Third, we derive one level-$\alpha_{k,j}$ linear boundary for each $k \in \NN, j \in \ZZ$ using the techniques of Step 2 so that $\sum_{k\in\NN}\sum_{j \in \ZZ} \alpha_{k,j} \leq \alpha$ and take a union bound over all of them. Here, $k \in \NN$ will index exponentially spaced epochs of time $t \in \NN$, while $j \in \ZZ$ will index evenly-spaced log-odds of $p \in (0, 1)$. Fourth, we modify the boundaries derived in Step 3 to obtain a boundary that is uniform in both $t \in \NN$ \emph{and} in $p \in (0, 1)$. Fifth and finally, we obtain an analytic upper bound on the boundary derived in Step 4.

At several points throughout the proof, we will make use of various functions that depend on $k$ and $j$. While we will define them as they are needed, we also list them here for reference.
\begin{subequations}
  \begin{align}
    W_t &:= \sum_{i=1}^t w_i^2, \label{eq:W_t}\\
    \alpha_{k,j} &:= \frac{\alpha}{(k+1)^s (|j| \lor 1)^s \zeta(s) (2\zeta(s) + 1)},\label{eq:alpha_kj}\\
    q(k,j) &:= \frac{1}{1+\exp \left \{ -2j\delta / \eta^{k/2} \right \}}, \label{eq:q_kj}\\
    j(k,p) &:= \left \lceil \frac{\eta^{k/2} \logit(p) }{2\delta} \right \rceil, \\
    \lambda(k,j) &:= \psi^{-1}_{G, q(k,j)}( \log(1/\alpha_{k,j}) / \eta^{k+1/2}),~~\text{where}~~\psi_{G,c}^{-1}(a) := \frac{2}{c + \sqrt{c^2 + 2/a}}, ~~\text{and}\label{eq:lambda_kj}\\
    b_{t,k}(p) &:= \frac{W_t \psi_{G, p} (\lambda_{k,j}) + \log(1/\alpha_{k,j})}{\lambda_{k,j}}.
  \end{align}
  \end{subequations}

    \paragraph{Step 1: Deriving an $e$-process.}
    Invoking \cref{lemma:bernsm} combined with \citet[Table 5, row 7]{howard_exponential_2018} we have that for any $p \in (0, 1)$, $S_t(p)$ is sub-gamma \citep{howard_exponential_2018,howard2018uniform} with variance process $V_t(p) := \sum_{i=1}^t ( w_i \1\{ R_i \leq Q(p)\} )^2$ and scale $c = p$, meaning we have that for any $\lambda \in [0, 1/c)$,
    \begin{equation}
      M_t^G(\lambda; p) := \exp\left \{ \lambda S_t(p) - V_t(p) \psi_{G, p}(\lambda) \right \}
    \end{equation}
    forms a test supermartingale.
    Now, since $V_t(p) \leq \sum_{i=1}^t w_i^2 \equiv W_t$ almost surely, we have that
    \begin{equation}
      E_t^G(\lambda; p) := \exp \left \{ \lambda S_t(p) - W_t \psi_{G,p}(\lambda) \right \} \leq M_t^G( \lambda; p)
    \end{equation}
    forms an $e$-process --- i.e.~it is upper-bounded by a test supermartingale.
    This completes the first step of the proof.

    \paragraph{Step 2: Applying Ville's inequality to $E_t^G(\lambda; p)$, yielding a time-uniform linear boundary.}
    In Step 1, we showed that $E_t^G(\lambda; p)$ forms an $e$-process. By Ville's maximal inequality for nonnegative supermartingales \citep{ville1939etude}, we have that
    \begin{equation}
      \PP(\exists t \in \NN : E_t^G( \lambda; p) \geq 1/\alpha) \leq \PP(\exists t \in \NN : M_t^G( \lambda; p) \geq 1/\alpha) \leq \alpha.
    \end{equation}
    Now, we will rewrite the inequality $E_t^G( \lambda; p) \geq 1/\alpha$ slightly more conveniently so that we can derive a time-uniform concentration inequality for $S_t(p)$. Indeed,
    \begin{align*}
      E_t^G( \lambda; p) \geq 1/\alpha &\iff \lambda S_t(p) - W_t \psi_{G,p}(\lambda) \geq \log(1/\alpha)\\
      &\iff S_t(p) \geq \underbrace{\frac{W_t \psi_{G,p}(\lambda) + \log(1/\alpha)}{\lambda}}_{b_{t}(p)}.
    \end{align*}
    In summary, we have the following time-uniform concentration inequality on $S_t(p)$ for any $p \in (0, 1)$, $\alpha \in (0, 1)$ and $\lambda \in [0, 1/p)$,
    \begin{equation}\label{eq:S_t-concentration-simple}
      \PP \left ( \exists t \in \NN : S_t(p) \geq b_t(p) \right ) \leq \alpha, ~~~\text{where}~~ b_t(p) := \frac{W_t \psi_{G,p}(\lambda)+ \log(1/\alpha)}{\lambda},
    \end{equation}
    which could also be written as a time-uniform high-probability upper bound on $S_t(p)$:
    \begin{equation}
      \PP \left ( \forall t \in \NN, \ S_t(p) < b_t(p) \right ) \geq 1-\alpha.
    \end{equation}
    
    \paragraph{Step 3: Union-bounding over infinitely many choices of $\lambda$, $\alpha$, and $p$.}
    In Step 2, we showed that $b_t(p)$ forms a time-uniform high-probability upper bound for $S_t(p)$. We will now take a union bound over a countably infinite two-dimensional grid of $t$ and $p$. Concretely, for each $k \in \NN$ and $j \in \ZZ$, recall $\alpha_{k, j}$, $q(k,j)$, and $\lambda(k,j)$ as in \eqref{eq:alpha_kj}, \eqref{eq:q_kj}, and \eqref{eq:lambda_kj}.
    The exact choices of $q(k,j)$ and $\lambda(k,j)$ will become relevant later. For now, note that by \eqref{eq:S_t-concentration-simple} from Step 2 combined with a union bound, we have that
    \begin{align}
      &\PP \left ( \exists t \in \NN, k \in \NN, j \in \ZZ : S_t(q(k,j)) \geq b_{t,k}(q(k,j)) \right ) \leq \sum_{k \in \NN} \sum_{j \in \ZZ} \alpha_{k, j}, \label{eq:union-bound-alpha_kj}\\
        &\text{where}~~~ b_{t,k}(q(k,j)) := \frac{W_t\psi_{G, q(k,j)}(\lambda_{k,j}) + \log(1/\alpha_{k,j})}{\lambda_{k,j}}.
    \end{align}
    We will now show that $\sum_{k\in \NN}\sum_{j \in \ZZ}\alpha_{k,j} = \alpha$ so that \eqref{eq:union-bound-alpha_kj} holds with probability at most $\alpha$. Indeed,
    \begin{align*}
      \sum_{k \in \NN} \sum_{j \in \ZZ} \alpha_{k,j} &= \frac{\alpha}{\zeta(s) (2\zeta(s) + 1)}\sum_{k \in \NN} \frac{1}{(k+1)^s} \sum_{j \in \ZZ} \frac{1}{(|j| \lor 1)^s} \\
      &= \frac{\alpha}{\zeta(s) (2\zeta(s) + 1)} \underbrace{\sum_{k =0} \frac{1}{(k+1)^s}}_{= \zeta(s)} \underbrace{\left ( 1 + 2 \sum_{m=1}^\infty \frac{1}{m^s} \right ) }_{= 2\zeta(s) + 1}\\
                                                     &= \alpha.
    \end{align*}
    Therefore, in summary, we have that
    \begin{equation}\label{eq:union-bound-alpha}
      \PP \left ( \exists t \in \NN, k \in \NN, j \in \ZZ : S_t(q(k,j)) \geq b_{t,k}(q(k,j)) \right ) \leq \alpha.
    \end{equation}

    \paragraph{Step 4: Removing dependence on $j \in \ZZ$ and obtaining $p$-uniformity.}
    We will obtain a bound that is uniform in $p \in (0, 1)$ by replacing $j$ with $j(k,p)$ --- a function of both $k \in \NN$ and $p \in (0, 1)$.
    For any $k \in \NN$ and any $p \in (0, 1)$, define $j(k, p)$ as
    \begin{equation}\label{eq:j_kp}
      j(k,p) := \left \lceil \frac{\eta^{k/2} \logit(p/(1-p)}{2\delta} \right \rceil.
    \end{equation}
    Of course, $j(k,p) \in \ZZ$ is not unique. It is easy to check that $p \leq q(k,j(k,p))$, a fact that we will use shortly. Abusing notation slightly, let $j_1, j_2, \dots$ denote the integers generated by $j(k,p)$ for every $k \in \NN$ and $p \in (0, 1)$, and let $\Jcal := \{j_1, j_2, \dots\} \subseteq \ZZ$ denote their image. Given this setup and applying \eqref{eq:union-bound-alpha} from Step 3, we have that
    \begin{align*}
      &\PP \left ( \exists t \in \NN, k \in \NN, p \in (0, 1) : S_t(q(k,j(k,p))) \geq b_{t, k}(q(k,j(k,p))) \right )\\
      =\ & \PP \left ( \exists t \in \NN, k\in \NN, j \in \Jcal : S_t(q(k,j)) \geq b_{t,k}(q(k,j))  \right ) \\
      \leq\ & \PP \left ( \exists t \in \NN, k\in \NN, j \in \ZZ : S_t(q(k,j)) \geq b_{t,k}(q(k,j))  \right )\\
      \leq\ & \alpha,
    \end{align*}
    where the second line follows from the definition of $\Jcal$, the third follows from the fact that $\Jcal \subseteq \ZZ$, and the last follows from \eqref{eq:union-bound-alpha}.
    In summary, we have the time- and $p$-uniform concentration inequality given by
    \begin{align}
      &\PP \left ( \exists t \in \NN, k \in \NN, p \in (0, 1) : S_t(q(k,j(k,p))) \geq b_{t, k}(q(k,j(k,p))) \right ) \leq \alpha,~~\text{or equivalently,}\\
      &\PP \left ( \forall t \in \NN, k \in \NN, p \in (0, 1),\  S_t(q(k,j(k,p))) < b_{t, k}(q(k,j(k,p))) \right ) \geq 1-\alpha \label{eq:t-and-p-uniform-Stpkjkp}
    \end{align}

    \paragraph{Step 5: Obtaining a time- and $p$-uniform upper bound on $S_t(p)$.}
    While \eqref{eq:t-and-p-uniform-Stpkjkp} is now written to be $p$-uniform, the quantity $b_{t,k}(q(k,j(k,p)))$ is only a high-probability upper bound on $S_t(q(k,j(k,p)))$, but what we need is a high-probability upper bound on $S_t(p)$. To this end, we use a similar technique to \citet{howard2022sequential} to bound the distance between $S_t(q(k,j(k,p)))$ and $S_t(p)$ for any $p \in (0, 1)$. Indeed, consider the following representation of $S_t(p)$ in terms of $S_t(q(k,j(k,p)))$:
    \begin{align*}
      S_t(p) &:= \sum_{i=1}^t w_i \1(R_i \leq Q(p)) - tp \\
             &\leq \sum_{i=1}^t w_i \1(R_i \leq Q(q(k,j(k,p)))) - tp\\
             &= S_t(q(k,j(k,p))) + t(q(k,j(k,p))-p),
    \end{align*}
    where the first line follows by definition of $S_t(p)$, the second by monotonicity of $Q \mapsto \1(R_t \leq Q)$ and the fact that $p \leq p_{k, j(k,p)}$, and the third follows from the definition of $S_t(q(k,j(k,p)))$. Combining~\eqref{eq:t-and-p-uniform-Stpkjkp} with the above representation of $S_t(p)$, we have that
    \begin{equation}
      \label{eq:upper-bound-on-S_t_p-forall-k}
      \PP \left ( \forall t \in \NN, k \in \NN, p \in (0, 1),\ S_t(p) < \underbrace{b_{t,k} (q(k,j(k,p)))}_{\mathrm{(i)}} + \underbrace{t(q(k,j(k,p)) - p)}_{\mathrm{(ii)}} \right ) \geq 1-\alpha,
    \end{equation}
    where $\mathrm{(i)} \equiv b_{t,k}(q(k,j(k,p)))$ is given by
    \begin{equation}
      b_{t,k} (q(k,j(k,p))) := \frac{W_t\psi_{G, q(k,j(k,p))}(\lambda_{k,j(k,p)}) + \log(1/\alpha_{k,j(k,p)})}{\lambda_{k,j(k,p)}}.
    \end{equation}

    \paragraph{Step 5(i): Upper-bounding $\mathrm{(i)}$ without dependence on $k$.}
        Applying \cref{lemma:btkj-upper-bound} but with $j(k, p)$ in place of $j$, we have that for every $\eta^k \leq W_t \leq \eta^{k+1}$,
    \begin{align*}
      b_{t,k}(q(k,j(k,p))) \leq \ &\sqrt{\gamma_1^2 \log(1/\alpha_{k,j(k,p)}) W_t + \gamma_2^2 q(k,j(k,p))^2 \log^2(1/\alpha_{k,j(k,p)})}\\
                                 \ &+ \gamma_2 q(k,j(k,p)) \log(1/\alpha_{k,j(k,p)}).
    \end{align*}
    Now notice that the above upper bound depends on $k$ solely through $q(k,j(k,p))$ and $\log(1/\alpha_{k,j(k,p)})$, each of which we will upper-bound independently of $k$. 
    By \cref{lemma:q_kj-upper-bound}, we have that
    \begin{equation}
      q(k,j(k,p)) \leq \bar q_t(p) \equiv \logit^{-1} \left ( \logit(p) + 2\delta \sqrt{\frac{\eta}{W_t}} \right ) ~~~\text{for all } \eta^k \leq W_t \leq \eta^{k+1},\\
    \end{equation}
    so it remains to upper-bound $\log(1/\alpha_{k,j(k,p)})$.
    Recall the definition of $\alpha_{k,j}$ for any $k \in \NN, j \in \ZZ$ given in \eqref{eq:alpha_kj}. Then we can write $\log(1/\alpha_{k,j(k,p)})$ as
    \begin{equation}\label{eq:log-1-by-alpha-kj}
      \log(1/\alpha_{k,j(k,p)}) = \underbrace{s\log(k+1)}_{(\star k)} + \underbrace{2\log(|j(k,p)| \lor 1)}_{(\star j)} + \log\zeta(s) + \log(2\zeta(s) + 1) + \log(1/\alpha),
    \end{equation}
    and we observe that $(\star k)$ and $(\star j)$ are the only terms depending on $k$. 
    Firstly, notice that $(\star j)$ can be upper bounded for every $\eta^k \leq W_t \leq \eta^{k+1}$ as
    \begin{align*}
      (\star j) \equiv 2\log(|j(k,p)| \lor 1) &= 2\log \left ( \left|\left \lceil \frac{\eta^{k/2}\logit(p)}{2\delta} \right \rceil \right| \lor 1\right ) \\
      &\leq 2\log \left ( \left |  \left \lceil \frac{\sqrt{W_t} \logit(p)}{2\delta} \right \rceil \right |  \lor 1\right ).
    \end{align*}
    Second, notice that we can easily upper-bound $(\star k)$ on epoch $\eta^k \leq W_t \leq \eta^{k+1}$ as
    \begin{align*}
      s\log(k+1) &\leq s\log \left ( \log_\eta W_t + 1\right ) .
    \end{align*}
    Therefore, we have the following upper-bound on $\log(1/\alpha_{k,j(k,p)})$:
    \begin{align*}
      \log(1/\alpha_{k,j}) \leq \ &s\log \left ( \log_\eta W_t + 1\right ) +
                                2\log \left ( \left |  \left \lceil \frac{\sqrt{W_t} \logit(p)}{2\delta} \right \rceil \right |  \lor 1\right )+
                                  \log\zeta(s) + \log(2\zeta(s) + 1) + \log(1/\alpha),\\
                                  & \equiv \ell_t(p),
    \end{align*}
    which no longer depends on $k$. In summary, we have that
    \begin{equation*}
      b_{t,k} (q(k,j(k,p))) \leq \sqrt{\gamma_1^2 \ell_t(p) W_t + \gamma_2^2 \widebar q_t(p)^2 \ell_t(p)^2} + \gamma_2 \widebar q_t(p) \ell_t(p).
    \end{equation*}

    \paragraph{Step 5(ii): Upper-bounding $\mathrm{(ii)}$ without dependence on $k$.}
    By \cref{lemma:q_kj-upper-bound}, we have that $q({k,j(k,p)}) \leq \widebar q_t(p) \equiv \logit^{-1} \left ( \logit(p) + 2\delta \sqrt{\frac{\eta}{W_t}} \right )$
    Therefore, we can upper bound $\mathrm{(ii)}$ as 
    \begin{align}
      \mathrm{(ii)} \equiv  t(q({k, j(k,p)})-p) &\leq t(\widebar q_t(p) - p) \equiv t \left [ \logit^{-1} \left ( \logit(p) + 2\delta \sqrt{\frac{\eta}{W_t}} \right ) - p \right ],
    \end{align}
    where the final inequality no longer depends on $k$.
    In sum, with probability at least $1-\alpha$, 
    \begin{align}
      \forall t \in \NN, p \in (0, 1),\ S_t(p) <\ &\sqrt{\gamma_1^2 \ell_t(p) W_t + \gamma_2^2 \widebar q_t(p)^2 \ell_t(p)^2} 
       + \gamma_2 \widebar q_t(p) \ell_t(p) + t(\widebar q_t(p) - p) .
    \end{align}

    \begin{lemma}\label{lemma:btkj-upper-bound}
      For any $k \in \NN$ and any $j \in \ZZ$, we have that for all $\eta^k \leq W_t \leq \eta^{k+1}$, 
      \begin{equation}
        b_{t,k}(q(k,j)) \leq \sqrt{\gamma_1^2 \log(1/\alpha_{k,j}) W_t + \gamma_2^2 q(k,j)^2 \log^2(1/\alpha_{k,j})} + q(k,j)\gamma_2 \log(1/\alpha_{k,j}),\label{eq:btkj-upper-bound-gammas}
      \end{equation}
      where $\gamma_1, \gamma_2$ are constants defined as
      \begin{equation}\label{eq:gammas}
        \gamma_1 := \frac{\eta^{1/4} + \eta^{-1/4}}{\sqrt{2}}~~\text{and}~~ \gamma_2 := \frac{\sqrt{\eta} + 1}{2}.
      \end{equation}
    \end{lemma}
    \begin{proof}
      Recall the chosen value of $\lambda_{k,j}$ given in \eqref{eq:lambda_kj},
      \begin{equation}
    \lambda(k,j) := \psi^{-1}_{G, q(k,j)}( \log(1/\alpha_{k,j}) / \eta^{k+1/2}),~~\text{where}~~\psi_{G,c}^{-1}(a) := \frac{2}{c + \sqrt{c^2 + 2/a}}
      \end{equation}
      Similar to \citet[Eq. (39)]{howard2018uniform}, some algebra will reveal that for any $t, k \in \NN, j \in \ZZ$, we have that
      \begin{equation*}
        b_{t,k}(q(k,j)) = A_{q(k,j)}\left ( \frac{\log(1/\alpha_{k,j})}{\eta^{k+1/2}} \right ) \underbrace{\left [ \sqrt{\frac{\eta^{k+1/2}}{ W_t}} + \sqrt{\frac{W_t}{\eta^{k+1/2}}} \right ]}_{(\star)} \sqrt{\frac{\log(1/\alpha_{k,j})W_t}{2}},
      \end{equation*}
      where $A_{c}(a) := \sqrt{2a} / \psi_{G,c}^{-1}(a) = \sqrt{1 + c^2a/2} + c\sqrt{a/2}$. 
      Now, notice that the second derivative of $(\star)$ with respect to $W_t$ is positive on $W_t \in [\eta^k, \eta^{k+1}]$, and hence $(\star)$ is convex in $W_t$. As such, for every $W_t \in [\eta^k, \eta^{k+1}]$ --- i.e. the $k^\mathrm{th}$ epoch --- we have that $(\star)$ is maximized at the endpoints $W_t = \eta^k$ and $W_t = \eta^{k+1}$, and we thus have the following upper bound on $b_{t,k}(q(k,j))$ on the $k^\mathrm{th}$ epoch:
      \begin{equation}\label{eq:btkj-upper-bound-kth-epoch}
        b_{t,k}(q(k,j)) \leq A_{q(k,j)}\left ( \frac{\log(1/\alpha_{k,j})}{\eta^{k+1/2}} \right ) \left [ \eta^{1/4} + \eta^{-1/4} \right ] \sqrt{\frac{\log(1/\alpha_{k,j})W_t}{2}}.
      \end{equation}
      Furthermore, since $W_t / \sqrt{\eta} \leq \eta^{k+1/2}$ on the $k^\mathrm{th}$ epoch, we also have that
      \begin{equation}\label{eq:upper-bound-K-A-function}
        A_{q(k,j)} \left ( \frac{\log(1/\alpha_{k,j})}{\eta^{k+1/2}} \right ) \leq A_{q(k,j)} \left ( \frac{\sqrt{\eta}\log(1/\alpha_{k,j})}{W_t} \right )~~~~\text{for all } \eta^k \leq W_t \leq \eta^{k+1}.
      \end{equation}
      Putting \eqref{eq:btkj-upper-bound-kth-epoch} and \eqref{eq:upper-bound-K-A-function} together, we have that for all $\eta^k \leq W_t \leq \eta^{k+1}$, 
      \begin{align}
        b_{t,k}(q(k,j)) &\leq \frac{\eta^{1/4} + \eta^{-1/4}}{\sqrt{2}} \left ( \sqrt{\log(1/\alpha_{k,j}) W_t + \frac{\sqrt{\eta}q(k,j)^2 \log^2(1/\alpha_{k,j})}{2}} + q(k,j)\frac{\eta^{1/4}\log(1/\alpha_{k,j})}{\sqrt{2}} \right )\nonumber \\
                        &= \sqrt{\gamma_1^2 \log(1/\alpha_{k,j}) W_t + \gamma_2^2 q(k,j)^2 \log^2(1/\alpha_{k,j})} + q(k,j)\gamma_2 \log(1/\alpha_{k,j}),
      \end{align}
      where $\gamma_1, \gamma_2$ are constants defined in~\eqref{eq:gammas}.
      This completes the proof of \cref{lemma:btkj-upper-bound}.
    \end{proof}

    \begin{lemma}\label{lemma:q_kj-upper-bound}
      Define $\widebar q_t(p)$ as
      \begin{equation}
        \widebar q_t(p) := \logit^{-1} \left ( \logit (p) + 2\delta \sqrt{\frac{\eta}{W_t}} \right ).  
      \end{equation}
      For all $\eta^k \leq W_t \leq \eta^{k+1}$, we have that
      \(
	q(k, j(k,p)) \leq \widebar q_t(p).
      \)
    \end{lemma}
    \begin{proof}
        The result follows by definition of $q(k, j(k,p))$. Indeed, we have that for all $\eta^k \leq W_t \leq \eta^{k+1}$,
        \begin{align*}
            q(k,j(k,p)) &:= \frac{1}{1+\exp \left \{ -2j(k,p)\delta / \eta^{k/2} \right \}}\\
            &= \left (1+\exp \left \{ -2\left \lceil \frac{\eta^{k/2} \logit(p) }{2\delta} \right \rceil\delta / \eta^{k/2} \right \} \right )^{-1} \\
            &\leq \left (1+\exp \left \{ -2 \left ( \frac{\eta^{k/2} \logit(p) }{2\delta} + 1\right )\delta / \eta^{k/2} \right \} \right )^{-1}\\
            &= \left (1+\exp \left \{ -(\logit(p) - 2\delta / \eta^{k/2}) \right \} \right )^{-1} \\
            &= \logit^{-1}(\logit (p) + 2\delta / \eta^{k/2})\\
            &\leq \logit^{-1}\left (\logit (p) + 2\delta \sqrt{\frac{\eta}{W_t}}\right ),
        \end{align*}
        which completes the proof.
    \end{proof}

\section{A causal view of contextual bandits via potential outcomes}\label{section:causal}

In \cref{section:introduction}, we discussed how the OPE problem can be interpreted as asking a \emph{counterfactual} question, such as ``how would the rewards have been, had we used a different policy $\pi$ than the logging policy $h$ that collected the data?''. While it is somewhat reasonable to think about the functional $\polval_t = \EE_\pi(R_t \mid X_1^{t-1})$ in a counterfactual sense, the Neyman-Rubin potential outcomes framework was designed for the rigorous study of precisely these types of causal questions \citep{splawa1990application,rubin1974estimating}. In this section, we will define a target \emph{causal} functional $\polval_t^\star$ in terms of potential outcomes, and outline the identification assumptions under which $\polval_t^\star$ is equal to $\polval_t$ (and hence, the conditions under which our \cs{}s can be interpreted as covering the causal quantity $\polvalt^\star := \frac{1}{t}\sum_{i=1}^t \polval_i^\star$). We emphasize that these identification assumptions are not required for the \cs{}s to be useful or sensible --- indeed, $\polvalt$ is still an interpretable statistical quantity that we may wish to estimate --- but they cannot otherwise be said to cover a causal functional defined in terms of potential outcomes.

Making our setup precise, we posit that for each time $t$, there is one potential outcome $R_t(a)$ for every action $a \in \Acal$. The functional we are ultimately interested in estimating is the \emph{conditional mean potential outcome reward under the policy $\pi$}, i.e.
\begin{align}
  \polval_t^\star \equiv \EE_\pi(R_t(G) \mid X_1^{t-1}) :=&\ \EE \left \{ \EE_{G \sim \pi(\cdot \mid X_t)}(R_t(G) \mid X_t, X_1^{t-1}) \mid X_1^{t-1} \right \}\\
  =&\ \int_{\Acal \times \Xcal} \EE(R_t(g) \mid G=g, X_t = x, X_1^{t-1}) \pi(g \mid x) p_{X_t}(x \mid X_1^{t-1}) \dd g \dd x.
\end{align}
In words, $\nu_t$ is the average of the potential outcomes $\{R_t(g)\}_{g \in \Acal}$ conditional on $X_1^{t-1}$ with respect to the distribution $\pi(\cdot \mid X_t)$. We use $g$ and $G$ in place of $a$ and $A$ to avoid confusion between the actual (random) action $A_t$ played according to the logging policy $h_t(\cdot \mid X_t)$ and the hypothetical (random) action $G$.
Without further assumptions, however, the counterfactual quantity $\polval_t^\star$ is not necessarily identified, meaning it cannot necessarily be written as a functional of the distribution of the observed data $\infseqt{X_t, A_t, R_t}$. This is simply due to the fact that the potential outcome $R_t(g)$ is not directly observable from $\infseqt{X_t, A_t, R_t}$. To remedy this, consider the following causal identification assumptions for every subject $t$,
\begin{enumerate}[start=1,label={(\IA{}\arabic*):},align=left]
\item \textbf{Consistency}: $A_t = a \implies R_t(a) = R_t$ for every $a \in \Acal$ with positive $\pi$-density,
\item \textbf{Sequential exchangeability}: $A_t \indep R_t(a) \mid X_1^t$, and
\item \textbf{Positivity}: $\pi \ll h_t$, meaning $h_t(A_t \mid X_t) = 0 \implies \pi(A_t \mid X_t) = 0$ almost surely.
\end{enumerate}
Notice that in the contextual bandit setup, \IA{2} and \IA{3} are known to hold \emph{by design}, while \IA{1} is more subtle (e.g.~\IA{1} may not hold even in a randomized experiment due to interference between subjects, such as in a vaccine trial). Nevertheless, with \IA{1}, \IA{2}, and \IA{3} in mind, we are ready to state the main identification result of this section.

\begin{lemma}\label{lemma:identification}
  Under causal assumptions \IA1, \IA2, and \IA3, we have that
  \begin{equation}
    \polval_t^\star = \polval_t, ~~\text{and hence}~~ \polvalt^\star := \frac{1}{t}\sum_{i=1}^t \polval_i^\star =  \frac{1}{t}\sum_{i=1}^t \polval_i =: \polvalt.
  \end{equation}
  In other words, the counterfactual conditional mean $\polval_t^\star$ can be represented as a function of the distribution of observed data $\infseqt{X_t, A_t, R_t}$, and that representation is given by $\polval_t$.
\end{lemma}
\begin{proof}
  The proof is an exercise in causal identification and essentially follows that of \citet[Theorem 1]{kennedy2019nonparametric} and Robins' $g$-formula \citep{robins1986new}, but we nevertheless provide a derivation here for completeness.

  Writing out the definition of $\polval_t^\star$, we have
  \begin{align}
    \polval_t^\star :=&\ \int_{\Acal \times \Xcal} \EE(R_t(g) \mid G=g, X_t = x, X_1^{t-1}) \pi(g \mid x_t) p_{X_t}(x \mid X_1^{t-1}) \dd g \dd x  \\ 
    =&\ \int_{\Acal \times \Xcal} \EE(R_t(g) \mid G=g, A_t = g, X_t = x, X_1^{t-1}) \pi(g \mid x) p_{X_t}(x \mid X_1^{t-1}) \dd g \dd x\label{eq:exchangeability} \\ 
    =&\ \int_{\Acal \times \Xcal} \EE(R_t(g) \mid A_t = g, X_t = x, X_1^{t-1}) \pi(g \mid x) p_{X_t}(x \mid X_1^{t-1}) \dd g \dd x\label{eq:PO-exchangeability}\\ 
             =&\ \int_{\Acal \times \Xcal} \EE(R_t \mid A_t = g, X_t = x, X_1^{t-1}) \pi(g \mid x) p_{X_t}(x \mid X_1^{t-1}) \dd g \dd x \label{eq:consistency}\\ 
    =&\ \EE \left \{ \EE_{A_t \sim \pi(\cdot \mid X_t)} (R_t \mid X_t, X_1^{t-1}) \mid X_1^{t-1} \right \}\\
    \equiv&\ \EE_\pi(R_t \mid X_1^{t-1}) =: \polval_t,
  \end{align}
  where \eqref{eq:exchangeability} follows from \IA{2} (sequential exchangeability), \eqref{eq:PO-exchangeability} follows from the fact that $R_t (G) \indep G \mid X_1^t$ (by definition), and \eqref{eq:consistency} follows from \IA{1} (consistency). Throughout, we implicitly used \IA{3} (positivity) so that the outer integral is well-defined. That is, we conditioned on $A_t = g$ at several points, which implicitly leaves us with a factor of $\pi(g \mid x) / h_t(g \mid x)$ --- positivity ensures that this quantity is well-defined with probability one. This completes the proof of \cref{lemma:identification}.

\end{proof}

\begin{remark}[On the relationship between OPE and stochastic intervention effect estimation]\label{remark:equivalence-OPE-stochastic-interventions}
  It is no surprise that the proof of \cref{lemma:identification} follows \citet{robins1986new} and \citet{kennedy2019nonparametric} who study \emph{stochastic intervention effects} in causal inference. Indeed, OPE and estimation of stochastic interventions are two different framings of essentially the same problem, and use the same importance-weighted and doubly robust estimators. The main differences between these fields lie in their emphases: the former is focused on adaptive experiments where logging policies are data-adaptive and \emph{known}, whereas the latter typically places more emphasis on potential outcomes and causal identification, observational studies (i.e. where logging policies must be \emph{estimated}), and more complex causal functionals, such as those of \citet{haneuse2013estimation} and \citet{kennedy2019nonparametric}. Of course, these are incomplete characterizations made with broad strokes; for a more detailed summary of prior work in stochastic interventions, see \citet[Section 1]{kennedy2019nonparametric}. 
\end{remark}

\begin{remark}[Implications for design-based confidence sequences]
  As an alternative to estimating treatment effects in superpopulations, one can opt to consider a so-called ``design-based'' approach to causal inference where the potential outcomes of all individuals are conditioned on, and confidence intervals are constructed for the \emph{sample} average treatment effect (SATE) given by $\mathrm{SATE}_t := \frac{1}{t}\sum_{i=1}^t (R_i(1) - R_i(0))$ where $R_i(a)$ is subject $i$'s potential outcome under treatment $a \in \{0, 1\}$. Here, the resulting confidence intervals cover the SATE with high probability, where the probability is taken with respect to the randomness in the treatment assignment mechanism only. The design-based approach goes back to Fisher and has a deep and extensive literature \citep{fisher1936design,splawa1990application,imbens2015causal}, and more recent work has constructed nonasymptotic \cs{}s for the time-varying effect $\infseqt{\mathrm{SATE}_t}$ in \citet[Section 4.2]{howard2018uniform} and asymptotic ones in \citet{ham2022design}. For a more comprehensive literature review, we direct readers to \citet{abadie2020sampling} and \citet{ham2022design} as well as the references therein.

  We simply remark here that the results of~\cref{section:time-varying} simultaneously apply to the design-based \emph{and} superpopulation settings as immediate corollaries. Indeed, in the stochastic (non-design-based) setting for binary experiments and under the causal identification assumptions \IA{1}--\IA{3}, we have that \cref{lemma:identification} yields
  \begin{equation}
    \poldifft^\star := \frac{1}{t}\sum_{i=1}^t \EE \left [ R_i(1) - R_i(0) \right ]  = \frac{1}{t}\sum_{i=1}^t \left [ \EE(R_i \mid A_i = 1) - \EE(R_i \mid A_i = 0)  \right ] =: \poldifft.
  \end{equation}

  Now, to recover \cs{}s for $\infseqt{\mathrm{SATE}_t}$ we simply condition on $\infseqt{R_t(1), R_t(0), X_t}$ so that $\infseqt{A_t}$ are the only non-degenerate random variables here. The techniques for time-varying treatment effects described in \cref{section:policy-value-differences} and \cref{section:ate-randomized-expts}, yield a $(1-\alpha)$-\cs{} $[ L_t , U_t]_{t=1}^\infty$ for $\infseqt{\poldifft} \equiv \infseqt{\poldifft^\star}$ and hence for $\infseqt{\mathrm{SATE}_t}$. Going further, when instantiated for the design-based setting, our \cs{}s substantially improve on \citet[Section 4.2]{howard2018uniform}, both practically and theoretically. Indeed, as discussed in~\citet[Section 3.2]{ham2022design}, one of the drawbacks of existing nonasymptotic \cs{}s in the literature is that the minimal propensity score --- i.e. $p_\mathrm{min} := \essinf_{t,a,x} \PP(A_t = a \mid x \mid \history_{t-1})$ --- must be specified in advance, and the downstream \cs{}s always scale with $p_\mathrm{min}^{-1}$. However, as we have emphasized throughout this paper, beginning with desideratum 5 in \cref{section:desiderata}, \emph{none} of our \cs{}s suffer from this limitation.

  Simultaneously, if we consider the superpopulation setting where $\EE(R_t(1)) - \EE(R_t(0)) = \delta$ for all $t \geq 1$ and for some $\delta \in [-1, 1]$, then under identification assumptions \IA{1}--\IA{3}, the same \cs{} $[L_t, U_t]_{t=1}^\infty$ also covers $\delta$ by \cref{lemma:identification}. In this way, our time-varying \cs{}s simultaneously handle the stationary superpopulation setting where treatment effects do not change over time, as well as the design-based setting where all potential outcomes are conditioned on, since these are both special cases of the time-varying stochastic setting considered in \cref{section:time-varying}.
\end{remark}


\end{document}